\documentclass[a4paper,11pt]{article}%
\usepackage{fullpage,tablefootnote}
\usepackage{amsfonts}
\usepackage{fancyhdr}
\usepackage{comment}
\usepackage{amsmath,amssymb,amsfonts}
\usepackage{graphicx}
\usepackage{float}
\usepackage{hyperref}
\usepackage[sort,comma]{natbib}
\usepackage{float}
\usepackage{algorithm}
\usepackage{algorithmic}
\usepackage{xcolor}
\usepackage{enumitem} % Tables & lists
\usepackage{booktabs}
\usepackage{multirow}
\usepackage{stackengine}
\usepackage{verbatim}

\newcommand{\RandChore}{\mathsf{RandChore}}
\newcommand{\RandMixed}{\mathsf{RandMixed}}

\newcommand{\OPT}{\mathsf{OPT}}

\newcommand{\cV}{\mathcal{V}}

\newcommand{\cI}{\mathcal{I}}
\newcommand{\cM}{\mathcal{M}}

\newcommand{\EW}{\mathsf{EW}}
\newcommand{\UW}{\mathsf{UW}}

\usepackage{amsthm}     % for using \theoremstyle{plain}, {definition}, {remark}
\theoremstyle{plain}   % default. it sets the text in italic and adds extra space above and below the
\newtheorem{theorem}{Theorem}[section]

\newtheorem{lemma}[theorem]{Lemma}
\newtheorem{proposition}[theorem]{Proposition}
\theoremstyle{definition} % adds extra space above and below, but sets the text in roman.
\newtheorem{definition}[theorem]{Definition}

\theoremstyle{remark} % Font is set in roman, with no additional space above or below.
\renewenvironment{proof}[1][Proof]{\noindent\textbf{#1.} }{\ $\Box$\medskip}

\usepackage{xcolor}

\begin{document}
\title{Randomized Strategyproof Mechanisms\\ with Best of Both Worlds Fairness and Efficiency}
\author{Ankang Sun\thanks{Department of Computing, Hong Kong Polytechnic University, China;
ankang.sun@polyu.edu.hk}
\and Bo Chen\thanks{Corresponding author: Warwick Business School, University of Warwick, UK;
bo.chen@wbs.ac.uk} }
\date{}

\maketitle

\begin{abstract}
We study the problem of mechanism design for allocating a set of indivisible items among agents
with private preferences on items. We are interested in such a mechanism that is strategyproof
(where agents’ best strategy is to report their true preferences) and is expected to ensure
fairness and efficiency to a certain degree. We first present an impossibility result that a
deterministic mechanism does not exist that is strategyproof, fair and efficient for allocating
indivisible chores. We then utilize randomness to overcome the strong impossibility. For
allocating indivisible chores, we propose a randomized mechanism that is strategyproof in
expectation as well as ex-ante and ex-post (best of both worlds) fair and efficient. For
allocating mixed items, where an item can be a good (i.e., with a positive utility) for one agent
but a chore (i.e., a with negative utility) for another, we propose a randomized mechanism that
is strategyproof in expectation with best of both worlds fairness and efficiency when there are
two agents.

\medskip\noindent\textbf{Key words:} multi-agent systems; resource allocation; mechanism design;
strategyproof; randomization
\end{abstract}

\section{Introduction}
Resource allocation is a fundamental issue that has obtained significant attention in the
intersection of operations research, economics and computer science
\citep{10.5555/49354,DBLP:reference/choice/2016,10.7551/mitpress/2954.001.0001}. Among different
types of resources, indivisible items stand out as a prominently studied model. This model has
drawn substantial interest as it serves as a metaphor of real-world allocation problems, such as
outpatient diagnostic tests scheduling \citep{ZAERPOUR2017283}, school choice
\citep{10.1257/000282803322157061}, order dispatching in ride-hailing platform
\citep{doi:10.1287/inte.2020.1047} and conference paper assignment
\citep{Lian_Mattei_Noble_Walsh_2018}.

In most practical allocation problems, individual preferences are private. This necessitates
allocation protocols to operate based on announced preferences. When participants are aware that
declared preferences determine the outcome of the allocation, strategic behavior often comes into
play, leading to misreports of true preferences for more favorable outcomes. False preferences
significantly undermine outcomes if the allocation protocol is not designed to address such
manipulation. Some recent studies have addressed such strategic behavior in supply chain management
\citep{KARAENKE2020897}, production system \citep{NORDE2016441} and project management
\citep{CHEN2021666}.

To circumvent the challenge of strategic manipulation, an ideal allocation protocol should
incentivize truthful reporting, making it the optimal strategy for individuals. Allocation
protocols designed with this property are known as \emph{strategyproof} mechanisms. After the
seminal work of \citet{NISAN2001166}, a line of works focus on the model of indivisible items and
strategyproof mechanisms (without money payment\footnote{When money transfer is either infeasible
or illegal in practice.}) to achieve fair outcomes \citep{bezakovaAllocatingIndivisibleGoods2005,
DBLP:conf/wine/MarkakisP11, amanatidisTruthfulMechanismsMaximin2016a,
amanatidisTruthfulAllocationMechanisms2017}. Many argue that, in general, strategyproofness
prohibits mechanisms from achieving fair or even approximately fair outcomes. In particular, for
indivisible goods, \citet{amanatidisTruthfulAllocationMechanisms2017} consider the notion of
envy-free up to one item\footnote{No one envies another after hypothetically eliminating one good
from the bundle of the latter.} (EF1) and show that EF1 and strategyproofness are incompatible even
for two agents and five items when agents' valuations are additive. To escape from the strong
impossibility result, follow-up studies have explored restricted preference domains, such as the
binary margin where the marginal valuation of an item is either zero or one. These studies
demonstrate that strategyproofness, fairness and efficiency can be achieved simultaneously under
the setting of binary margin \citep{DBLP:conf/aaai/BabaioffEF21, DBLP:conf/aaai/BarmanV22,
halpernFairDivisionBinary2020}.

The aforementioned positive results are mainly established in the context of allocating goods,
items with positive utilities. However, in practical scenarios, the items to be allocated extend
beyond goods to chores, which impose costs on individuals. The chore allocation scenario includes,
but is not limited to, workload assignment in manufacturing industry \citep{PEREIRA2023387}, job
scheduling \citep{doi:10.1287/opre.2022.2321} and household chores division among family members
\citep{Igarashi_Yokoyama_2023}. Moreover, whether an item is a good or a chore is not only
subjective but also agent dependent. For example, a pesticide may have a positive utility for a
farmer who wants to protect his crops from pests, but a negative utility for a beekeeper who loses
his bees due to the toxic chemicals. Such a setting is referred to as \emph{mixed items}. For
chores or mixed items, whether strategyproofness, efficiency and fairness can be achieved
simultaneously is not well understood. In particular, an approach applied successfully in the
setting of goods does not convert to chores (or mixed items) straightforwardly. Halpern et
al.~\citep{halpernFairDivisionBinary2020} propose a fair, efficient and strategyproof mechanism for
indivisible goods and additive valuations with a binary margin. The proposed mechanism relies on
the allocation that maximizes the product of individuals' utilities, known as Nash welfare. The
allocation with the maximum Nash welfare allocation is EF1 and Pareto optimal for goods
\citep{caragiannisUnreasonableFairnessMaximum2019}, while no (known) equivalent rule with fairness
and efficiency guarantees for chores. This leads to our first research question:

\begin{quote}
For allocating chores or mixed items, whether strategyproofness, fairness and efficiency are
compatible? In general or in a restricted preference domain?
\end{quote}

Restricting preference domains is one approach to escape from strong impossibilities. Another
potential avenue involves relaxing the requirement of (deterministic) strategyproofness to
strategyproofness \emph{in expectation} through lotteries. Informally, a randomized mechanism is
considered strategyproof in expectation if every agent's expected utility is maximized when they
truthfully reveal their preferences. For allocating indivisible items, achieving performance
guarantees before realization (ex-ante) along with strategyproofness in expectation is not a hard
task. For instance, consider a mechanism that selects an agent uniformly at random and assigns all
items to that agent. This mechanism is strategyproof in expectation as reported preferences are
ignored. In the allocation of such a mechanism, no agent prefers the bundle of another before
realization. However, such a mechanism unlikely guarantees fairness or efficiency after realization
(ex-post). An ideal randomized mechanism should have both ex-ante and ex-post performance
guarantees, also referred to as the ``best of both worlds''
\citep{hyllandEfficientAllocationIndividuals1979, DBLP:conf/wine/BabaioffEF22,
freemanBestBothWorlds2020, akrami2023randomized}. This then leads to our second research question:

\begin{quote}
For allocating chores or mixed items, whether strategyproofness in expectation, best of both
worlds fairness and efficiency are compatible? In general or in a restricted preference domain?
\end{quote}

\subsection{Main results}

This paper presents results in the realm of fair and efficient allocation mechanisms, particularly
focusing on the allocation of (i) \emph{indivisible chores}, where all items are chores for agents,
and (ii) \emph{mixed items}, where an item can be a good for one agent but a chore for another. In
Section~\ref{sec::deterministic}, we study deterministic mechanisms and present two impossibility
results. We first examine the performance of the widely-studied mechanism, known as {sequential
picking}, revealing that this allocation rule does not guarantee Pareto optimality. Then we
establish an impossibility result for all deterministic mechanisms, which informally states that
deterministic strategyproofness, equitability-based fairness notions and Pareto optimality are
incompatible, even for additive valuations with a binary margin. Given this strong impossibility
result, we relax deterministic strategyproofness to strategyproofness in expectation and
concentrate on randomized mechanisms in the subsequent sections.

In Section~\ref{Sec::chores}, we consider randomized mechanisms and are concerned with allocation
of indivisible chores. We focus on the \emph{$k$-restricted additive} valuations, where a single
item has $k$ inherent values and the item's value for an agent is either from the $k$ inherent
values or $0$. For $1$-restricted additive valuations, we propose a randomized mechanism that is
group-strategyproof\,\footnote{No group of agents has an incentive to misreport.} in expectation
and has best of both worlds fairness and efficiency guarantees. We also show that fairness and
efficiency guarantees of the proposed mechanism cannot be achieved simultaneously by any mechanism
(without the requirement of strategyproofness) when valuations are $k$-restricted additive for all
$k \geq 2$. Moving on to Section~\ref{Sec::MixedItems}, we consider the setting of mixed items and
propose a randomized mechanism for two agents that is strategyproof in expectation. The proposed
mechanism also guarantees best of both worlds fairness and efficiency.

\subsection{Other related works}

The study of strategyproof mechanisms with guaranteed properties dates back to
\citet{hyllandEfficientAllocationIndividuals1979}. Then investigating the existence of
strategyproof mechanisms with desired properties was a central concern. \citet{gale1987college}
raises the question of whether we can find a ``nice'' mechanism that satisfies strategyproofness,
efficiency, and other distributional properties, simultaneously. A follow-up study
\citet{zhouConjectureGaleOnesided1990} partially answers Gale's conjecture and proves that in the
problem of assigning $n$ items to $n$ agents, no mechanism can be strategyproof, Pareto optimal and
{symmetric}\footnote{In a symmetric mechanism, agents with the same reported valuations should
receive identical (expected) utility.} when $ n \geq 3$. Then
\citet{papaiStrategyproofSingleUnit2001} studies the problem of allocating a single indivisible
item to several agents without money transfer and examines possible extra properties that a
strategyproof mechanism can have. In another study, \citet{papaiStrategyproofNonbossyMultiple2001}
considers the scenario where each agent can receive more than one item and presents a
characterization of subclass of strategyproof mechanisms.
\citet{bezakovaAllocatingIndivisibleGoods2005} consider the problem of allocating $m$ items to two
agents and show that no strategyproof mechanism can maximize the value of the worse agent.
\citet{DBLP:conf/wine/MarkakisP11} present a lower bound on the values of agents in the worst-case
scenario and also show that no deterministic strategyproof mechanism can achieve a
$(2/3)$-approximation of the worst-case bound.

As indicated in the aforementioned studies, pursuing an optimal strategyproof mechanism is
sometimes impossible. \citet{procacciaApproximateMechanismDesign2013} start to consider
strategyproof mechanisms that achieve an approximation of the optimal objective when the optimal
solution is computationally intractable or when no such mechanism exists.
\citet{amanatidisTruthfulMechanismsMaximin2016a} focus on indivisible goods and present a
deterministic strategyproof mechanism with $O(m)$-approximation on maximin-share fairness. They
complement the result by showing that no deterministic strategyproof mechanism can achieve better
than $(1/2)$-approximation. For allocating indivisible chores, \citet{wu-add1} assume ordinal
valuations for agents and present deterministic and randomized mechanisms that return allocations
guaranteeing $O(\log(m/n))$-approximation and $O(\log(\sqrt{n}))$-approximation of maximin share,
respectively.

There are also some studies concerning strategyproofness along with other fairness criteria such as
EF1. \citet{amanatidisTruthfulAllocationMechanisms2017} demonstrate that no deterministic
strategyproof mechanism guarantees EF1, even there are only two agents with additive valuations. To
achieve positive results, subsequent studies focus on domains of restricted agent preferences.
\citet{halpernFairDivisionBinary2020} focus on additive valuations with a binary margin and
indicate that the rule of maximizing Nash welfare with a lexicographic tie-breaking rule is group
strategyproof, EF1 and Pareto optimal. \citet{DBLP:conf/aaai/BabaioffEF21} consider submodular and
monotone valuation functions with a binary margin and show that the Lorenz dominating\footnote{For
two ascending-ordered vectors $\mathbf{a}$ and $\mathbf{b}$, the Lorenz domination partial order is
defined as $\mathbf{a} \succ_{Lorenz} \mathbf{b}$ if for every $k$, the sum of the first $k$
entries of $\mathbf{a}$ is at least as large as that of $\mathbf{b}$. {A Lorenz dominating
allocation is the allocation whose valuation vector Lorenz dominates the vector of every other
allocation.}} rule is strategyproof, EF1 and Pareto optimal. A follow-up work
\citep{DBLP:conf/aaai/BarmanV22} show that the proposed mechanism in
\citet{DBLP:conf/aaai/BabaioffEF21} is indeed group strategyproof.

\subsection{Organization of the paper}

After introduction of basic definitions in Section~\ref{sec:preliminaires}, we establish our
impossibility results in Section~\ref{sec::deterministic} for deterministic mechanisms. We then
follow up with randomization and present our randomized mechanisms, respectively for allocating
chores and mixed items in Sections~\ref{Sec::chores} and \ref{Sec::MixedItems}, together with their
nice properties. Proofs of lemmas, propositions and theorems are provided either in the main body
of the paper or are relegated to the appendix. We make some concluding remarks in
Section~\ref{sec:conclusions}.

%%%%%%%%%%%%%%%%%%%%%%%%%%%%%%%%%%%%%%%%%%%%%%%%%%%%%%%%%%%%%%%%%%%%%%%%
\section{Preliminaries}
\label{sec:preliminaires}

Denote by $[k] = \{1, \ldots, k \}$ for any $ k \in \mathbb{N}^+$. A set $E = \{ e_1,\ldots, e_m
\}$ of $m$ indivisible items are to be allocated to a set $N = \{1, \ldots, n \}$ of $n$ agents.
Each agent $ i $ is associated with a valuation function $ v_i : 2^E \rightarrow \mathbb{R}$. An
item $e$ is a good (resp., a chore) for agent $ i $ if $v_i(e) \geq 0 $ (resp., $v_i(e) \leq 0$).
Throughout this paper, we are particularly interested in two settings: (i) chores: all items are
chores for all agents, and (ii) mixed items: an item can be a good for agent $i$ but a chore for
some agent $ j \neq i $. Agents are rational and can interact strategically towards making a
collective decision. Let $\mathbb{A}$ be the set of \emph{outcomes} or \emph{allocations}. An
allocation $\mathbf{A} = (A_1,\ldots, A_n) \in \mathbb{A}$ is an $n$-partition of $E$, that is, $A_
i \cap A_  j = \emptyset$ for any $i, j \in [n]$ with $ i \neq j $ and $\bigcup_{ i \in [n]} A_i =
E$. Given an allocation $\mathbf{A} = (A_1,\ldots, A_n)$, for any $i \in [n]$, $A_i$ represents the
set of items or the \emph{bundle} allocated to agent $i$. Note that in any allocation, no item is
left unallocated.

Throughout the paper, for any $i \in [n]$, function $ v_i $ represents the true valuation function
of agent $i$, and following the convention of mechanism design literature, $v_i$ is also called the
\emph{type} of agent $i$. Each agent $i$ privately observes his preferences over outcomes in
$\mathbb{A}$, which is defined by that agent $i$ knows his valuation function $ v _i $ that is not
available to others. Denote by $V_i$ the set of all possible types of agent $i$ and by $\cV =
V_1\times \cdots \times V_n$ the set of type profiles. A type profile is denoted by $\mathbf{v} =
(v_1,\ldots, v_n)$. Additionally, each agent $i$ is associated with a utility function $u_i :
\mathbb{A} \times V _i \rightarrow \mathbb{R}$. Given an outcome $\mathbf{A}$ and type $v_i \in V_i
$, utility $ u_i (\mathbf{A}, v_i)$ denotes the value or payoffs of agent $i$ under outcome
$\mathbf{A}$ when his type is $ v_i $. In this paper, we focus on the setting where agent $i$ does
not care about how bundles are allocated to others, i.e., \emph{no externalities}. Hence, the
utility of agent $ i $ with type $v_i$ under allocation $\mathbf{A} = (A_1,\ldots, A_n)$ can also
be expressed as $u_i ( A_  i , v_i )$, equivalent to $ v_i ( A_ i )$.

The \emph{social choice function} or \emph{mechanism} is a mapping $\cM : V_1 \times \cdots \times
V _n \rightarrow \mathbb{A}$ that assigns every possible type profile an outcome from set
$\mathbb{A}$. We focus on \emph{direct} mechanisms that elicit type information from agents by
asking them to reveal their true valuations. After receiving reported profile $\hat{\mathbf{v}} =
(\hat{v}_1, \ldots, \hat{v}_n )$, mechanism $\cM$ outputs an allocation based on
$\hat{\mathbf{v}}$. Given the combinatorial explosion, it is unrealistic for agents to report their
valuations on every subset $T \subseteq E$. Consequently, we assume that valuation function $v_i$
is \emph{additive} for all $i \in [n]$, that is, for any $T\subseteq E$, $ v_i ( T ) = \sum_{e \in
T } v_i ( \left\{  e \right\})$. For simple notations, we use $v_i (e)$, instead of $ v_i ( \{ e
\})$, to denote agent $i$'s valuation of item $e$. The set $\mathbb{A}$ of outcomes, the set $N$ of
agents, the set $E$ of items and the type sets $ \{ V_i \}_{ i \in [n]} $ are common knowledge,
while type $v_i$ is private to agent $i$ only.

\subsection{Fractional and randomized allocations}

In this paper, we also write allocations in the form of matrices. With a slight abuse of notation,
an allocation or an \emph{allocation matrix} $ \mathbf{A} = ( a _{j,i})_{ j \in[m], i\in [n]}$ is
such that a fraction $a_{j , i }$ of item $e_j$ allocated to agent $i$ and $\sum_{ i \in [n]} a_{ j
, i } = 1 $ for all $ j \in [m]$. A \emph{deterministic} allocation requires $a_{ j , i } \in \{ 0
,1 \}$ for all $ j , i $, while in a \emph{fractional} allocation $\mathbf{A}$, $a_{ j , i }$ can
be any real number in interval $[0,1]$. A randomized allocation $\widetilde{\mathbf{A}}$ is a
probability distribution over deterministic allocations, and moreover, it naturally implements a
fractional allocation. More specifically, let the support of $\widetilde{\mathbf{A}}$ be
$\mathbf{A}^1, \ldots, \mathbf{A}^k \in \{0,1\}^n$ and the probability of being $\mathbf{A} ^ j $
is $ p _j $ for all $ j \in [k]$. Then we say that $\widetilde{\mathbf{A}}$ \emph{implements}
fractional allocation $ \sum_{ j \in [k]} p_j \mathbf{A} ^ j $. Note that for a given fractional
allocation, there can be several randomized allocations that implement it.

Given an arbitrary randomized allocation $\widetilde{\mathbf{A}}$, let us suppose it implements
fractional allocation $\mathbf{A} ^ {\prime} = ( a ^ {\prime}_{j , i })_{ j \in [m], i \in [n]}$.
The term $a ^{\prime}_{ j , i }$ can be interpreted as the probability of assigning item $e_j$ to
agent $ i $. For any $ i \in [n]$, the expected value of agent $i$ under $\widetilde{\mathbf{A}}$
is $\sum_{ j \in [m]} a_{ j , i} ^ {\prime} v_ i ( e_j )$, which is equivalent to the value of
agent $ i $ under fractional allocation $\mathbf{A} ^ {\prime}$.

\subsection{Strategyproofness}

Let us formally define the notions of (group) strategyproofness and strategyproofness in
expectation. We first introduce a common notation used in the mechanism design literature. Given
the reported profile $\hat{\mathbf{v}}$ and a set of agents $S \subseteq N$, denote by $\hat{v}_S$
the set of reported valuations of agents $S$ and by $\hat{v}_{ -S }$ the set of reported valuations
of agents $N\setminus S$, i.e., $ \hat{v}_{-S} = \hat{v}_{ N \setminus S}$. Similar notations
extend to agent type and the set of agent types.

\begin{definition} [SP]
A deterministic mechanism $\mathcal{M}$ is strategyproof if $ \forall i\in [n]$, $ \forall v_i \in
V_i $, $ \forall \hat{v}_i \in V_i $, $\forall \hat{v}_{ - i } \in V_{ - i }$, the following holds:
$$
u_i (\mathcal{M}({v}_i , \hat{v}_{ - i }), v_i ) \geq 	u_i ( \mathcal{M}(\hat{v}_i ,
\hat{v}_{ - i }), v_i ).
$$
\end{definition}

In a strategyproof (SP) mechanism, it is always an optimal strategy for an agent $i$ to truthfully
report his type, regardless of what is reported by other agents. In other words, no agent can gain
additional utility by misreporting compared to the received utility when reporting truthfully. A
mechanism is \emph{randomized} if it returns randomized allocations. Accordingly, we use the notion
of strategyproof in expectation (SPIE) for a randomized mechanism under which every agent's
expected utility is maximized when reporting truthfully. Formally, we have the following
definition.

\begin{definition}[SPIE]\label{Definition::SPIE}
A randomized mechanism $\widetilde{\mathcal{M}}$ is strategyproof in expectation if $ \forall i\in
[n]$, $ \forall v_i \in V_i $, $ \forall \hat{v}_i \in V_i $, $\forall \hat{v}_{ - i } \in V_{ - i
}$, the following holds,
$$
\mathbb{E}_{\mathbf{A}\thicksim \widetilde{\cM}(v_i, \hat{v}_{-i})}[u_i (\mathbf{A}, v _i )] \geq
\mathbb{E}_{\mathbf{A}\thicksim \widetilde{\cM}(\hat{v}_i, \hat{v}_{-i }) }[u_i (\mathbf{A}, v_i )].
$$
\end{definition}

SP and SPIE mechanisms can incentivize truthful reporting for individual agents but may not prevent
strategic manipulation from a group of agents. To address strategic behavior at the group level
there is a need for \emph{group-strategyproof} (GSP) and \emph{group-strategyproof in expectation}
(GSPIE) mechanisms, in which no group of agents can collude to misreport their valuations in a way
that makes at least one member of the group better off without making any of the remaining members
worse off.
\begin{definition}[GSP]
A deterministic mechanism $\mathcal{M}$ is group-strategyproof if there does not exist a group of
agents $S \subseteq N$ and a reported profile $(\hat{v}_S, \hat{v}_{ - S })$ such that for any $ i
\in S $,
$$ u_i (\mathcal{M}(\hat{v}_S, \hat{v}_{-S }),v_i)\geq u_i(\mathcal{M}({v}_S,\hat{v}_{-S}),v_i),$$
and at least one strict inequality holds.
\end{definition}

\begin{definition}[GSPIE]\label{Definition::GSPIE}
A randomized mechanism $\widetilde{\mathcal{M}}$ is group-strategyproof in expectation if there
does not exist a group of agents $S \subseteq N$ and a reported profile $(\hat{v}_S, \hat{v}_{ -
S})$ such that for any $ i \in S$,
$$
\mathbb{E}_{\mathbf{A}\thicksim \widetilde{\cM}(\hat{v}_S, \hat{v}_{ - S }) }[u_i (\mathbf{A}, v_i)]
\geq \mathbb{E}_{\mathbf{A}\thicksim \widetilde{\cM}({v}_S, \hat{v}_{ - S }) }[u_i (\mathbf{A}, v_i )],
$$
and at least one strict inequality holds.
\end{definition}

\subsection{Efficiency, fairness and other distributional properties}

We conclude this section by presenting efficiency measurement and fairness criteria that we are
concerned with in this study. We begin with {Pareto optimal} allocations and social welfare
functions.

\begin{definition}[PO]\label{def::PO}
An allocation $\mathbf{B} = (B _1,\ldots, B_n)$ Pareto dominates another allocation $\mathbf{A} =
(A_1,\ldots, A_n)$ if for any $ i \in [n]$, $ v_i ( B_i ) \geq v_i ( A_  i)$ and at least one
inequality is strict. An allocation $\mathbf{A}$ is Pareto optimal if no allocation Pareto
dominates it.
\end{definition}

\begin{definition}
Given an allocation $\mathbf{A}= ( A_1,\ldots, A_n)$, its utilitarian welfare is $\UW(\mathbf{A}) =
\sum _{ i \in [n]} v_i ( A_ i )$ and the egalitarian welfare is $\EW(\mathbf{A}) = \min_{ i \in
[n]} v_i ( A_  i)$. An allocation $\mathbf{A}$ is a utilitarian welfare maximizer
\textnormal{(UWM)} and an egalitarian welfare maximizer \textnormal{(EWM)} if $\mathbf{A}$ achieves
the maximum utilitarian welfare and egalitarian welfare, respectively.
\end{definition}

\begin{definition}[EF and EF1]
An allocation $\mathbf{A} = ( A_1,\ldots, A_ n)$ is envy-free \textnormal{(EF)} if for any $i,j\in
[n]$, $ v_i ( A_  i ) \geq v_i ( A_  j )$. An allocation $\mathbf{A}$ is envy-free up to one item
\textnormal{(EF1)} if for any $ i ,j \in [n]$, there exists $e \in A_i\cup A_j$ such that $ v_i (
A_ i \setminus \{ e \}) \geq v_i ( A_j \setminus \{ e \})$.
\end{definition}

Informally, for mixed items, an allocation is EF1 if for every agent $i$, he does not envy agent
$j$ after either  hypothetically removing a chore (from agent $i$'s view) in his own bundle or
removing a good in agent $j$'s bundle.

\begin{definition}[EQ and EQ1]
An allocation $\mathbf{A} = (A_1,\ldots, A_n)$ is equitable \textnormal{(EQ)} if for any $i,j \in
[n]$, $ v_i ( A_  i) = v_j ( A_  j )$. The allocation $\mathbf{A}$ is equitable up to one item
\textnormal{(EQ1)} if for any $ i , j \in [n]$, there exists $e \in A_i \cup A_j$ such that $ v_i (
A_  i \setminus \{ e \})  \geq v_j ( A_  j\setminus \{ e \})$.
\end{definition}

The notion of EQ1 requires that the valuation of agent $i$ must be at least that of agent $j$ when
either a chore (from agent $i$'s view) is removed from agent $i$'s bundle or a good is removed from
agent $j$'s bundle.

\begin{definition}[PROP and PROP1]
An allocation $\mathbf{A} = ( A_  1, \ldots, A_n)$ is proportional \textnormal{(PROP)} if for any $
i \in [n]$, $ v_i ( A_ i) \geq \frac{1}{n} v_i ( E )$. An allocation $\mathbf{A}$ is proportional
up to one item \textnormal{(PROP1)} if for any $ i \in [n]$, $ v _i  ( A_i  ) \geq \frac{1}{n} v_i
( E )$ or $ v_i ( A_ i \cup \{  e \} ) \geq \frac{1}{n} v_i ( E )$ for some $ e \in E \setminus
A_i$ or $ v_i ( A_i \setminus \{ e  \}) \geq \frac{1}{n}v_i ( E )$ for some $e \in A_i$.
\end{definition}

In a PROP allocation, every agent $i$ is guaranteed a value of $\frac{1}{n} v_i (E)$, also known as
his \emph{fair share}. In a PROP1 allocation, agent $ i $ is guaranteed his fair share after either
hypothetically receiving an extra good, not assigned to his yet, or eliminating a chore from his
bundle. We remark that in the setting of mixed items, when agents are associated with additive
valuations, EF allocations are PROP, and EF1 allocations are also PROP1
\citep{DBLP:journals/aamas/AzizCIW22}.

\section{Impossibility results of deterministic mechanisms}
\label{sec::deterministic}

In this section, we focus on deterministic mechanisms. Consider a class of widely-studied
mechanisms, called \emph{sequential picking} \citep{kohlerClassSequentialGames1971}, where agents
are ordered in advance and each agent $i$ picks a number $t_i $ of items in his turn. The advantage
of sequential picking is that once $\{ t_i \}_{ i \in [n]}$ are predetermined or unrelated to
reported profiles and $\sum_{i\in[n]} t_i = m $, the corresponding mechanism is strategyproof and
outputs a complete allocation. The sequential picking rule has been employed to design truthful
mechanisms with performance guarantees regarding share-based fairness notions
\citep{amanatidisTruthfulMechanismsMaximin2016a, wu-add1}. However, we find that for indivisible
chores, no sequential picking mechanism can always return Pareto optimal outcomes, even when the
margin of chores is binary, i.e., $ v_i ( e_j ) \in  \{ 0, -1 \}$ for all $i$ and $j$.

\begin{theorem}\label{proposition::SEQ-pick}
For indivisible chores, no sequential picking mechanism can always return Pareto optimal
allocations, even when agents' valuations have a binary margin.
\end{theorem}

We next present a stronger impossibility result regarding all deterministic mechanisms in the
context of allocating indivisible chores.

\begin{theorem}\label{thm:: impossbility-deterministic-EQ1}
For indivisible chores, no deterministic mechanism can be \textnormal{SP}, \textnormal{PO} and
\textnormal{EQ1}, simultaneously, even when there are four chores and two agents with binary
additive valuations.
\end{theorem}

\begin{proof}
For a contradiction, let $\mathcal{M}$ be a
deterministic SP, PO and EQ1 mechanism.
Consider an instance $ \cI $ with two agents and a set $E$ of four chores $e_1, e_2, e_3, e_4$.
Given the reported profile $\mathbf{\hat{v}} = ( \hat{v}_1, \hat{v}_2 )$ with $ \hat{v}_i ( e_j ) = -1$ for all $i, j$,
mechanism $\mathcal{M}$ has to assign each agent two chores
so that the outcome can be EQ1 when agents report truthfully.
Without loss of generality, we assume $\mathcal{M}(\mathbf{\hat{v}}) = \mathbf{A}$ with
$ A_1 = \{ e_{1}, e_{ 2 }\}$
and $ A_2 = \{ e _{ 3}, e_{  4}\}$.
We now consider another reported profile $\mathbf{\hat{v}} ^{ \prime} = (\hat{v}^{\prime}_1, \hat{v} ^ {\prime}_2 )$ where $ \hat{v}^{\prime}_2 (e) = \hat{v}_2 (e)$ for all $e$ and
$ \hat{v} ^ {\prime}_1 ( e_{ 1 }) = \hat{v} ^ {\prime}_1 ( e_{ 2 }) = 0$ and $ \hat{v} ^ {\prime}_1 ( e_j ) = -1$ for $ j = 3, 4$.
Suppose $\mathcal{M}(\mathbf{\hat{v}} ^ {\prime}) = \mathbf{A} ^ {\prime}$.
As $\cM$ always returns PO allocations,
allocation $\mathbf{A}^{\prime}$ needs to be PO when $\mathbf{\hat{v}^{\prime}}$
is true valuations.
Thus,
$e_1, e_2$ should be assigned to agent 1 in $\mathbf{A}^{\prime}$.
Moreover, by the property of EQ1,
each agent $i \in [2]$ should be allocated one chore from $\{ e_3, e_4 \}$ under $\mathbf{A}^{\prime}$.

Now suppose
for each $i\in [2]$,
$ \hat{v} ^ { \prime }_i  $ is the type of agent $i$, i.e., $ \hat{v} ^ {\prime}_i = v_i $.
Then, if both agents report truthfully, $\cM$ returns allocation $\mathbf{A} ^ {\prime}$
with $ v_1 ( {A} _1 ^ {\prime} ) = - 1$.
However, if agent 1 misreports $\hat{v}_1$,
then the outcome becomes
$\mathcal{M}(\hat{v}_1, \hat{v}^{\prime}_2) = \mathbf{A}$,
and the value of agent 1 in the allocation $\mathbf{A}$ is equal to
$ v_1 ( A_1) = 0 $.
Therefore, agent 1 has incentive to behave strategically, contradicting the strategyproofness of $\mathcal{M}$.
\end{proof}

We remark that the above impossibility result is mainly owing to the fact that deterministic
strategyproofness is too strict to be achieved together with Pareto efficiency and fairness. Note
that for indivisible chores, EQ1 and PO allocations are guaranteed to exist for additive valuations
\citep{freemanEquitableAllocationsIndivisible2020b}. Given such a strong impossibility result, we
relax strategyproofness via randomness.

%%%%%%%%%%%%%%%%%%%%%%%%%%%%%%%%%%%%%%%%%%%%%%%%%%%%%%%%%%%%%%%%%%%%%%%%%%%%%%%%
%%%%%%%%%%%%%%%%%%%%%%%%%%%%%%%%%%%%%%%%%%%%%%%%%%%%%%%%%%%%%%%%%%%%%%%%%
\section{A randomized mechanism for indivisible chores}
\label{Sec::chores}

Throughout this section, we focus on the setting of indivisible chores, and assume agents have
$k$-\emph{restricted additive} valuation functions, that is, every chore $e_j$ is associated with
$k$ inherent values $ v^1(e_j), v^2(e_j),\ldots, v^k(e_j)$ with $v^t(e_j) < 0 $ for all $t \in
[k]$, and for any agent $ i $, his value of $e_j$ is $ v_i ( e_j ) \in \{ 0, v^1(e_j),
v^2(e_j),\ldots, v^k(e_j) \}$. In particular, 1-restricted additive valuation is simply called
\emph{restricted additive} valuation where for simple notations, the inherent value of $e_j$ is $v
( e_j ) < 0$. The restricted additive valuation has been studied in the literature on fair division
and resulted in significant and challenging research questions
\citep{bezakovaAllocatingIndivisibleGoods2005,
asadpourSantaClausMeets2012,DBLP:conf/ijcai/AkramiRS22}. Note that restricted additive functions
strictly generalize binary additive functions. Problems that can be solved efficiently in binary
additive valuations possibly become computationally intractable under the restrictive additive
domain. For example, in the context of indivisible goods, the egalitarian welfare-maximizing
allocation can be computed in polynomial time under binary additive valuations
\citep{10.5555/3237383.3237392}, while for restricted additive valuations, it is NP-hard to
approximate within a factor better than $1/2$ \citep{bezakovaAllocatingIndivisibleGoods2005}. As
type sets $\{ V_i \}_{ i \in [n]}$ are common knowledge, the inherent values $ \{ v ^ t ( e_j ) \}$
of $e_j $ is publicly known, which further constrains every agent $i$'s reported value on $e_j$ to
satisfy $\hat{v}_i( e_j) \in \{ 0, v^1(e_j), v^2(e_j), \ldots, v^k(e_j) \}$.

Designing SPIE mechanisms with fairness and efficiency guarantee is not a trivial task. Existing
randomized mechanism such as Random Priority and Probabilistic Serial
\citep{bogomolnaiaNewSolutionRandom2001} can not guarantee strategyproofness and fairness at the
same time when $ m > n $. While some straightforward mechanisms, such as assigning all items to an
agent chosen uniformly at random satisfies strategyproofness and fairness before realization, its
performance guarantees is poor after realization, with all items allocated to a single agent.
Therefore, among randomized mechanisms, we pursue those that can be strategyproof in expectation
and have best of both worlds performance guarantees. In the following, we define ex-ante and
ex-post fairness/efficiency.

\begin{definition}
Given fairness or efficiency notion $P$, a randomized allocation $\widetilde{\mathbf{A}}$ is
ex-ante $P$ if the fractional allocation implemented by $\widetilde{\mathbf{A}}$ is $P$ and is
ex-post $P$ if every deterministic allocation in its support is $P$.
\end{definition}

\begin{definition}
Given fairness or efficiency notions of $P_1$ and $P_2$, a randomized mechanism
$\widetilde{\mathcal{M}}$ is ex-ante $P_1$ and ex-post $P_2 $ if it always returns a randomized
allocation that is ex-ante $P_1$ and ex-post $P_2 $.
\end{definition}

For ex-ante and ex-post efficiency, we have the following implication relation.

\begin{proposition}\label{prop::implication-relation}
An ex-ante \textnormal{UWM} randomized allocation $\widetilde{\mathbf{A}}$ is also ex-ante
\textnormal{PO} and ex-post \textnormal{UWM}. An ex-ante \textnormal{PO} randomized allocation
$\widetilde{\mathbf{B}}$ is also ex-post \textnormal{PO}.
\end{proposition}

\begin{proof}
Suppose that $\widetilde{\mathbf{A}}$ has support $\{\mathbf{A}^1, \ldots, \mathbf{A} ^ k\}$ with
probability $p_i$ on deterministic allocation $\mathbf{A} ^i$ for each $ i \in [k]$. Let
$\mathbf{A} ^ {\prime}$ be the fractional allocation implemented by $\widetilde{\mathbf{A}}$. If
there exists another fractional allocation $\widehat{\mathbf{A}}$ that Pareto dominates
$\mathbf{A}^{\prime}$, then $\UW(\widehat{A}) > \UW(\mathbf{A} ^ {\prime})$ holds, contradicting
the fact that $\widetilde{\mathbf{A}}$ is ex-ante UWM. Hence, allocation $\widetilde{\mathbf{A}}$
is also ex-ante PO. Next, we prove that $\widetilde{\mathbf{A}}$ is also ex-post UWM. For a
contradiction, suppose $\widetilde{\mathbf{A}}$ is not ex-post UWM, then there must exist a
deterministic allocation $\mathbf{A} ^ {k+1}$ such that $\UW(\mathbf{A} ^ {k+1}) > \UW(\mathbf{A} ^
i )$ for some $ i \in [k]$. Consider the fractional allocation $\mathbf{A} ^ {\prime\prime} = p_i
\mathbf{A}^{k+1} + \sum_{ j \neq i } p_j \mathbf{A} ^ j $ and it is not hard to verify
$\UW(\mathbf{A} ^ {\prime\prime}) > \UW(\mathbf{A} ^ {\prime})$, contradicting the fact that
$\widetilde{\mathbf{A}}$ is ex-ante UWM. Therefore, $\widetilde{\mathbf{A}}$ is also ex-post UWM.

For an arbitrary ex-ante PO randomized allocation $\widetilde{\mathbf{B}}$, let its support be
$\{\mathbf{B}^1, \ldots, \mathbf{B} ^ k\} $ with probability $ p _i $ on deterministic allocation
$\mathbf{B} ^ i $ for each $ i \in [k]$. Moreover, let $\mathbf{B} ^ {\prime}$ be the fractional
allocation implemented by $\widetilde{\mathbf{B}}$. For a contradiction, assume that there exists
another deterministic allocation $\mathbf{B} ^ { k + 1 }$ that Pareto dominates $\mathbf{B} ^ i $
for some $ i \in [k]$. Similarly, the fractional allocation $\mathbf{B} ^ {\prime\prime} = p_i
\mathbf{B} ^ {k+1} + \sum_{ j \neq i } p_j \mathbf{A} ^ j $ Pareto dominates $\mathbf{B} ^
{\prime}$, a contradiction. Therefore, allocation $\widetilde{\mathbf{B}}$ is ex-post PO.
\end{proof}

In contract to UWM and PO, not every ex-ante fair solution guarantees ex-post fairness. Let us
again consider the mechanism of assigning all items to an agent chosen uniformly at random. It is
not hard to verify that the returned allocation is ex-ante EF. However, it provides little fairness
guarantee after realization.

The main result of this section is a GSPIE randomized mechanism (see
Algorithm~\ref{alg::randomchore}) with best of both worlds fairness and efficiency guarantees for $
n $ agents with 1-restricted additive valuation functions. Intuitively, the mechanism first
collects reported type profile and then partitions $[m]$ into $Q$ and $\bar{Q}$ where for any $e_j$
with $ j \in Q$, there exists some agent with reported value zero on $e _j $, while for any $e_{ j
^ {\prime}}$ with $ j ^ {\prime} \in \bar{Q}$, every agent $i$ reports $\hat{v}_i (e_{ j ^
{\prime}}) = v(e_{ j ^ {\prime}}) $. Then, for every $ j \in Q $, $\RandChore$ assigns $e_j$ to an
agent chosen uniformly at random with reported value $0$ on $e_j$. Items $\bigcup_{ j \in \bar{Q}}
e_ j$ are then assigned to agents in a round-robin fashion based on a permutation $\sigma$ of $\{
1, \ldots, n\}$ generated uniformly at random.

\begin{algorithm}[ht!]
\caption{$\RandChore$}
\begin{algorithmic}[1]\label{alg::randomchore}
\STATE Collects reported profile $(\hat{v}_1, \ldots, \hat{v}_n)$.
\STATE Let ${Q} = \{ q \in [m] \mid \exists i \textnormal{ such that } \hat{v}_i ( e_q ) = 0 \} $ and $\bar{{Q}} = [m] \setminus {Q}$.
  \label{ALG::GSP-CHORE-STEP-2}
\STATE For every $q\in{Q}$, uniformly at random pick an agent reporting zero value on $e_q$ and assign $e_q$
  to that agent. \label{ALG::GSP-CHORE-SETP-3}
\STATE Let $\bar{{Q}} = \{ l_1, l_2, \dots, l_k \}$ with inherent values $ v ( e_{ l_1 }) \geq \cdots \geq v ( e_{ l_ k } )$.
\STATE Uniformly at random generate a permutation $\sigma$ of $\{ 1,\ldots, n \}$. According to $\sigma$,
  assign the chore with the largest inherent value from the remaining items to an agent each time, until
  all chores are assigned. If there is a tie on the largest value item, pick $e_{ l_j }$ with the smallest
  index $j$. \label{ALG::GSP_CHORE-STEP-5}
\end{algorithmic}
\end{algorithm}

\begin{theorem}\label{thm::main-RandChore}
Mechanism $\RandChore$ is \textnormal{GSPIE} and satisfies the following:
\begin{itemize}
\item Fairness guarantee: ex-ante \textnormal{EF}, \textnormal{PROP}, \textnormal{EQ} and ex-post
    \textnormal{EF1}, \textnormal{PROP1}, \textnormal{EQ1};
\item Efficiency guarantee: ex-ante \textnormal{PO}, \textnormal{UWM}, \textnormal{EWM} and
    ex-post \textnormal{PO}, \textnormal{UWM}, $2$-approximation of \textnormal{EWM}.
\end{itemize}
\end{theorem}

In the following, we split the proof of Theorem~\ref{thm::main-RandChore}. We begin with GSPIE and
then prove best of both worlds fairness and efficiency. For simplicity, we use notation
$\widetilde{\cM} ^ *$ and $\RandChore$ interchangeably in this section.

\begin{proposition}\label{Prop::SP-RES-OPT}
Mechanism $\RandChore$ is \textnormal{GSPIE}.
\end{proposition}

Before proving Proposition~\ref{Prop::SP-RES-OPT}, we first present several lemmas.

\begin{lemma}\label{Lemma::GSP1}
Given the reported profile $(\hat{v}_1,\ldots, \hat{v}_n)$ and the corresponding $Q$ and $\bar{Q}$,
for any agent $i \in [n]$, he receives an expected utility $\frac{1}{n} v_i (\bigcup_{q \in
\bar{Q}} e_q )$ in Step~\ref{ALG::GSP_CHORE-STEP-5}.
\end{lemma}

\begin{proof}
Fix index $i$. Since $\sigma$ is a uniformly random permutation of $\{ 1,2,\ldots, n \}$, the
probability of agent $i $ on position $\sigma(j)$ is $\frac{1}{n}$ for all $j \in [n]$. Based on
Step~\ref{ALG::GSP_CHORE-STEP-5}, if agent $i$ is in position $\sigma(j)$, he receives a utility $
v_i (e_{ l_j } \cup  e_{l_{ n + j } } \cup \cdots \cup e_{l_{k n + j }})$, where $k=\lfloor
{(|\bar{{Q}}| - j )}/{n} \rfloor$. Thus, his expected utility derived by the assignment of
Step~\ref{ALG::GSP_CHORE-STEP-5} is equal to
\[
\sum_{ j = 1} ^ {n} \frac{1}{n} v_i  (\bigcup_{p =0 }^{k} e _{ l_{pn+j}})
 =  \frac{1}{n} v_i ( \bigcup _{j = 1} ^ n \bigcup_{p =0 }^{k} e _{ l_{pn+j}})
 = \frac{1}{n} v_i(\bigcup_{q \in \bar{{Q}}} e_q),
\]
where the first equality transition is due to the additivity of $v_i(\cdot)$.
\end{proof}

For any group $S$ of agents and any agent $i \in S$, if agent $i$'s true valuation on $e_j$ is $
v_i ( e_j ) = 0 $, then he cannot gain additional expected utility by misreporting $\hat{v}_i (e_j
) = v ( e_j )$.

\begin{lemma}\label{Lemma::GSP_chore2}
Given a subset $S\subseteq N $ and a reported profile $(\hat{v}_S , \hat{v}_{ - S })$ with
$\hat{v}_S \neq v_S $, construct another reported profile $( \hat{v} ^ {\prime}_S, \hat{v}_{ - S})$
as follows: for any $i \in S$ and $e_j \in E$, if $ v_i ( e_j ) = 0 $, then set $\hat{v}^{\prime}_i
( e_j ) = 0$; otherwise, $\hat{v}^{\prime}_i ( e_j ) = \hat{v}_i ( e_j )$. Then for any $i \in S$,
the expected utility of agent $i $ under reported profile $(\hat{v}^{\prime}_S, \hat{v}_{ - S })$
is at least that under $(\hat{v}_S, \hat{v}_{ - S })$,
$$
\mathbb{E }_ {\mathbf{A}\thicksim \widetilde{\cM}^* (\hat{v}^{\prime}_S, \hat{v}_{ - S})} [ u_i (\mathbf{A},  v_i )] \geq
\mathbb{E }_ {\mathbf{A}\thicksim \widetilde{\cM}^* (\hat{v}_S, \hat{v}_{ - S})} [ u_i (\mathbf{A},  v_i )].
$$
\end{lemma}

\begin{lemma}\label{Lemma::GSP_chore3}
Given a subset $ S \subseteq N $ and reported valuations $ \hat{v}_{ - S } $, the summation of the
expected utilities of agents in $S$ is maximized when every agent $i \in S $ reports his true
valuation.
\end{lemma}

\begin{proof}[Proof Sketch]
The proof is given by a contradiction argument. If some valuation functions $\hat{v}_S$ result in
the summation of expected utilities of agents in $S$ being larger than that under $v_ S $, then
according to Lemma~\ref{Lemma::GSP_chore2}, we can further assume $\hat{v}_S$ only contains one
type of misreporting: some agent $ j \in S $ reports $ \hat{v}_j ( e ) = 0 $ on chore $e$ while his
true valuation is $ v_j ( e ) = v  ( e ) < 0 $. Based on Steps~\ref{ALG::GSP-CHORE-STEP-2} and
\ref{ALG::GSP_CHORE-STEP-5}, deviating from $v(e)$ to zero can only increase the probability of
allocating $e$ to agents in $S$. Therefore, agents of $S$ have no incentive to manipulate. The
formal proof is deferred to Appendix~\ref{Appendix::Sec:chores}.
\end{proof}

Now we are ready to prove that $\RandChore$ is GSPIE.

\smallskip

\begin{proof}[Proof of Proposition~\ref{Prop::SP-RES-OPT}]
For the sake of a contradiction, assume there exists a group of agents $S \subseteq N$ and a
reported profile $(\hat{v}_S, \hat{v}_{ -S})$ such that for any agent $ i \in S$, it holds that
$$
\mathbb{E}_{ \mathbf{A} \thicksim \widetilde{\mathcal{M}}^*(\hat{v}_S, \hat{v}_{ - S })}
[ u_i ( \mathbf{A}, v_i )] \geq \mathbb{E}_{ \mathbf{A} \thicksim \widetilde{\mathcal{M}}^*(v_S, \hat{v}_{ - S })}
[ u_i ( \mathbf{A}, v_i )],
$$
and at least one strict inequality holds. Accordingly, we have the following inequality,
$$
\sum_{ i \in S} \mathbb{E}_{ \mathbf{A} \thicksim \widetilde{\mathcal{M}} ^*( \hat{v}_S,
\hat{v}_{ - S })} [ u_i ( \mathbf{A}, v_i )] > \sum_{ i \in S }\mathbb{E}_{ \mathbf{A}
\thicksim \widetilde{\mathcal{M}} ^* (v_S, \hat{v}_{ - S })} [ u_i ( \mathbf{A}, v_i )],
$$
which contradicts Lemma~\ref{Lemma::GSP_chore3}.
\end{proof}

After establishing the group strategyproofness, we proceed to prove that $\RandChore$ can ensure
best of both worlds fairness. Specifically, $\RandChore$ always returns solutions with ex-ante
exact fairness (EF, PROP and EQ) and ex-post approximate fairness (EF1, PROP1 and EQ1).

\begin{proposition}\label{prop::ex-post and ex-ante fairness of M}
Mechanism $\RandChore$ is ex-ante \textnormal{EF}, \textnormal{EQ} and \textnormal{PROP}, and
ex-post \textnormal{EF1},  \textnormal{EQ1} and \textnormal{PROP1}.
\end{proposition}

We then show that $\RandChore$ also guarantees efficiency from the best of both worlds perspective.

\begin{proposition}\label{prop::bobw-efficiency}
Mechanism $\RandChore$ is ex-ante \textnormal{PO}, \textnormal{UWM} and \textnormal{EWM} and
ex-post \textnormal{PO}, \textnormal{UWM} and $2$-approximation of \textnormal{EWM}.
\end{proposition}

Note that the $2$-approximation of optimal egalitarian welfare is almost the (ex-post) limitation
of $\RandChore$, especially when the number of agents is large.

\begin{proposition}\label{prop::impo-ewm}
Mechanism $\RandChore$ is not ex-post $(\frac{2n-1}{n} -\epsilon)$-approximation of
$\textnormal{EWM}$ for any $\epsilon>0$.
\end{proposition}

As for the running time of the proposed mechanism, although $\RandChore$ is a randomized mechanism,
randomness here does not affect the running time of execution. Note that the worst-case running
time of Step~\ref{ALG::GSP-CHORE-SETP-3} is $O(m)$ and of Step~\ref{ALG::GSP_CHORE-STEP-5} is $ O
(mn)$. Thus, the worst-case running time of $\RandChore$ is $O(mn)$.

We conclude this section by a discussion of possible extension of $\RandChore$. As demonstrated,
$\RandChore$ can guarantee EF1, EQ1, PO and strategyproofness simultaneously when agents have
restricted additive valuation functions. A natural question arises: in a broader preference domain,
such as additive valuation functions, can a strategyproof (in expectation) mechanism achieve EF1,
EQ1, and PO? Unfortunately, the answer is negative, even without the requirement of
strategyproofness. \citet{freemanEquitableAllocationsIndivisible2020b} have shown that for
indivisible chores, EF1, EQ1 and PO are incompatible for four agents with additive valuation
functions. For completeness, we also provide the example in Appendix~\ref{Appendix::Sec:chores}.
Note that in the example of \citet{freemanEquitableAllocationsIndivisible2020b}, valuation
functions of agents are indeed $2$-restricted additive, and consequently, the proposed mechanism
$\RandChore$ achieves the best that we can hope for in the $k$-restricted additive preference
domain.

%%%%%%%%%%%%%%%%%%%%%%%%%%%%%%%%%%%%%%%%%%%%%%%%%%%%%%%%%%%%%%%%%%%%%%%%%%%%%%%%%%%%%%%

\section{A randomized mechanism for allocating mixed items}
\label{Sec::MixedItems}

In this section, we consider allocation of mixed items to two agents, where an item can be a good
for an agent but a chore for another. For valuation functions, agents are assumed to have
\emph{M-restricted additive} valuations, which can be viewed as a generalization of restricted
additive functions to the setting of mixed items. Namely, each item $e_q $ has two inherent values
$\{ -c(e_q), v(e_q) \}$ with $c(e_q), v(e_q)>0$ and, if $e_q$ is a chore to (resp., a good for) an
agent, then his valuation of $e_q$ is $-c(e_q)$ (resp., $v(e_q)$). Otherwise, his valuation for
$e_q$ is equal to $0$. Note that $c(e_q)$ and $v(e_q)$ are not required to be identical. Recall
that type sets $\{ V_i \}_{ i \in [n]}$ are common knowledge, and hence for any $q\in [m]$,
inherent values $-c(e_q)$ and $v(e_q)$ are publicly known so that the reported valuation should be
$ \hat{v}_i ( e_q ) \in \{ -c(e_q), 0, v(e_q) \}$ for all $ i \in [n]$.

The main result of this section is a randomized SPIE mechanism with best of both worlds fairness
and efficiency for two agents. We propose mechanism $\RandMixed$ and formally introduce it below
(see Algorithm~\ref{ALG::mixed}). A simple description of $\RandMixed$ is as follows. Based on
reported profile $(\hat{v}_1, \hat{v} _2)$, partition $[m]$ into four parts $\{Q_0, Q_1, Q_2,
Q_3\}$. Assign items in $Q_0$ one-by-one to one of the two agents uniformly at random. For any $ i
\in \{1,2\}$, assign items $\bigcup_{q \in Q_i } e_q$ to agent $i$ and then assign items
$\bigcup_{q \in Q_3 } e_q$ to agents in a round-robin fashion based on a permutation $\sigma$
generated uniformly at random.

\begin{algorithm}[ht]\caption{$\RandMixed$}
\begin{algorithmic}[1]\label{ALG::mixed}
   \STATE Collect reported profile $(\hat{v}_1, \hat{v}_2 )$.
   \STATE Partition $[m] = Q_0 \cup Q_1 \cup Q_2 \cup Q_3 $ where $ Q_0 = \{ q \in [m] \mid \hat{v}_1
    (e_q ) = \hat{v}_2 (e_q) = 0 \}$, $Q_1 = \{ q \in [m] \mid \hat{v}_1 ( e_q ) > \hat{v}_2 (e_q) \}$,
    $Q_2 = \{ q \in [m] \mid \hat{v}_1 (e_q) < \hat{v}_2 (e_q) \}$ and $ Q_3 = \{ q \in [m]
    \mid \hat{v}_1 (e) = \hat{v}_2 (e) \neq 0  \}$. \label{two-players-mechanism-step-2}
\STATE For each $e_q$ with $ q \in Q_0 $, uniformly at random pick an agent and assign $e_q$ to that agent.
\STATE For $i =1, 2$, assign items $\bigcup_{ q \in Q_i} e_q $ to agent $i$.\label{two-players-mechanism-step-4}
\STATE Let $\sigma$ be a permutation of $\{1,2\}$ generated at random uniformly.
       Among unassigned items, assign the one with the largest inherent value to the two agents in a
       round-robin fashion based on permutation $\sigma$. \label{two-players-mechanism-step-5}
\end{algorithmic}
\end{algorithm}
	
\begin{theorem}\label{thm::mechanism-two-players}
Mechanism $\RandMixed $ is \textnormal{SPIE}, ex-ante \textnormal{PO}, \textnormal{UWM},
\textnormal{EF}, \textnormal{PROP} and ex-post \textnormal{PO}, \textnormal{UWM}, \textnormal{EF1}
and \textnormal{PROP1}.
\end{theorem}

In what follows, we split the proof of Theorem~\ref{thm::mechanism-two-players}, and for
simplicity, we use $\widetilde{\cM} ^ 2$ and $\RandMixed$ interchangeably in this section.

\begin{proposition}\label{prop::EV-RR-MIX}
Given reported profile $(\hat{v}_1, \hat{v}_2)$ and the corresponding partition $\{Q_i\}_{ i=0}^3$,
any agent $i\in \{1,2\}$ receives an expected utility $ \frac{1}{2} v_i (\bigcup_{ q \in Q _0 \cup
Q_3} e_q ) + v_i ( \bigcup_{ q \in Q_i } e_q) $.
\end{proposition}

\begin{proof}
Fix $i\in \{1,2\}$. As items $\bigcup _{ q \in Q_{3-i}} e_q $ are allocated to agent $3-i$, the
utility of agent $ i $ indeed comes from the assignment of $\bigcup_{ q \in Q _0 \cup Q_i \cup Q_3
} e_q$. Note that $e_q$ with $q \in Q_  0$ is assigned to agent $i $ with probability $
\frac{1}{2}$, and $e_q$ with $q \in Q_i$ is assigned to agent $i $ with probability $1$. As for
items $ \bigcup  _{ q \in Q_3} e_q $, similar to the analysis in the proof of
Lemma~\ref{Lemma::GSP1}, this part results in an expected utility $\frac{1}{2} v_i ( \bigcup _{q
\in Q_3} e_q)$ for agent $i$. Therefore, agent $i$'s expected utility derived from the assignment
is $ \frac{1}{2} v_i (\bigcup_{ q \in Q _0 \cup Q_3} e_q ) + v_i ( \bigcup_{ q \in Q_i } e_q) $.
\end{proof}
	
\begin{lemma}\label{Lemma::mechanism-two-players-1}
Given agent $i$ and reported profile $(\hat{v} ^k_i , \hat{v}_{ 3 - i })$ with $ \hat{v} ^ k_i (
e_k ) \neq v_i ( e_k )$ for some $k\in [m]$, construct another reported profile $(\hat{v}_i ,
\hat{v}_{ 3 - i })$ where $ \hat{v}_i( e_k) = v_i ( e_k )$ and $\hat{v}_i ( e ) = \hat{v} ^k_i(e)$
for all $ e \neq e_k$. Then, the expected utility of agent $i$ under reported profile $(\hat{v}_i,
\hat{v}_{3-i})$ is at least that under $(\hat{v}^k_i, \hat{v}_{3-i})$.
\end{lemma}

Intuitively, Lemma~\ref{Lemma::mechanism-two-players-1} states that from any non-truthful reported
profile, correcting the reported value of a single item of agent $i$ does not decrease the expected
utility of agent $i$. Then, for any reported valuation $\hat{v}_i \neq v_i $, one can start from
$\hat{v}_i$ and reach $v_i $ by a sequence of corrections of the reported values for a single item,
without decreasing the expected utility.

\begin{proposition}
Mechanism $\RandMixed$ is \textnormal{SPIE}.
\end{proposition}

\begin{proof}
Fix $ i\in \{1,2\}$ and $ \hat{v}_{ 3 - i}$. Let $\hat{v}_i $ be arbitrary reported valuations that
differ from the true valuations $ v_i $ on a set $\{ e_{p_1}, \ldots, e_{p_r }\}$ of $ r \in
\mathbb{N}^+ $ items. First let $ \hat{v}_i ^ 0 = \hat{ v }_i$ and then for every $ l \in [r]$,
construct valuation function $ \hat{v}_i ^ l $ as follows:
$$
\hat{v}_i^l\left(e\right) =
\begin{cases}
{v}_i(e_{p_l}) & \textnormal{ if } e = e_{ p_l }, \\
\hat{v}_i^{l-1}(e) & \textnormal{ otherwise}.
\end{cases}
$$
It is not hard to verify that valuation functions $ \hat{v}_i ^ r$ is identical to $  v_i $. Then,
according to Lemma~\ref{Lemma::mechanism-two-players-1}, we have
$$
\mathbb{ E }_{ \mathbf{A} \sim \widetilde{\cM}^2( \hat{v}^r_i, \hat{v}_{3-i})} [ u_i (\mathbf{A}, v_i )] \geq
\mathbb{ E }_{ \mathbf{A} \sim \widetilde{\cM}^2( \hat{v}^{r-1}_i, \hat{v}_{3 - i})} [ u_i (\mathbf{A}, v_i )] \geq \cdots \geq
\mathbb{ E }_{ \mathbf{A} \sim \widetilde{\cM}^2( \hat{v}^0_1, \hat{v}_{3-i})} [ u_i (\mathbf{A}, v_i )].
$$
Since $\hat{v}_i^r = v_i$ and $\hat{v}_i^0 = \hat{v}_i$, the above inequalities become
$$
\mathbb{E}_{ \mathbf{A} \sim \widetilde{\cM}^2( v_i, \hat{v}_{3 - i })} [ u_i ( \mathbf{A}, v_i )]
\geq \mathbb{ E }_{ \mathbf{A} \sim \widetilde{\cM}^2(\hat{v}_i, \hat{v}_{ 3 - i })} [ u_i (\mathbf{A}, v_i )],
$$
which completes the proof.
\end{proof}

In the following, we show that $\RandMixed$ outputs allocations that guarantee ex-ante and ex-post
fairness and efficiency.

\begin{proposition}\label{prop::two-players-fairness}
Mechanism $\RandMixed$ is ex-ante \textnormal{EF}, \textnormal{PROP} and ex-post \textnormal{EF1},
\textnormal{PROP1}.
\end{proposition}

\begin{proposition}\label{prop::two-players-efficiency}
Mechanism $\RandMixed$ is ex-ante and ex-post \textnormal{PO} and \textnormal{UWM}.
\end{proposition}

Since randomness does not affect the running time of execution, the worst-case running time of
$\RandMixed$ is $O(m)$. One may observe that mechanism $\RandMixed$ does not provide good
performance guarantee regarding the notion of equitability. We remark that this is due to the fact
that (relaxed) equitability and PO is incompatible in the mixed items setting even without the
requirement of strategyproofness. For example, consider the following instance of two agents and
two mixed items, in which Agent 1 values each item at $1$ and agent 2 values each item at $-1$. The
PO allocation assigns both items to agent 1, violating EQ and EQ1.

%%%%%%%%%%%%%%%%%%%%%%%%%%%%%%%%%%%%%%%%%%%%%%%%%%%%%%%%%%%%%%%%%%%%%%%%%%%%%%%%%%%%%%%%%%%%%%%%%%%%%

\section{Conclusions}
\label{sec:conclusions}

In this paper, we have studied allocation of indivisible items from the mechanism design
perspective. We have focused on the settings where (i) all items are chores and (ii) an item can be
a good for one agent and a chore for another. If randomization is not allowed, we have showed that
strategyproofness, fairness and efficiency are incompatible, and in particular, no deterministic
mechanism can be SP, PO and EQ1 simultaneously, even when agents' valuations are binary additive.
On the other hand, if randomization is allowed, we have proposed a GSPIE mechanism that guarantees
ex-ante EF, EQ and ex-post EF1, EQ1 and ex-ante UWM and EWM, when all items are chores and agents'
valuations are 1-restricted additive. We have also studied the model of mixed items and designed a
randomized SPIE mechanism with best of both worlds fairness and efficiency for two agents with
M-restricted additive valuations.

Our findings have significant implications for both the theoretical understanding and practical
implementation of fair and efficient allocation mechanisms. Our constructive proofs and analyses
provide a solid foundation for designing strategyproof mechanisms in complex allocation settings,
offering insights into the mechanisms’ ability to align agents’ incentives with truthful reporting.
Moreover, the proposed mechanisms offer a novel approach to overcoming the inherent challenges of
strategyproof mechanism design in the context of indivisible items allocation, showcasing the
potential of randomised mechanisms in achieving strategyproofness, fairness and efficiency
simultaneously. In practice, our proposed mechanisms can be implemented easily and efficiently in
various real-world scenarios where agents' preferences are private but fair and efficient resource
allocation is crucial.

Looking forward, there are several interesting open questions that deserve investigation. Firstly,
it remains to be seen whether the idea of $\RandMixed$ can be generalized to accommodate three or
more agents for mixed items. The second open question is if it is possible to achieve
strategyproofness in expectation and best of both worlds fairness and efficiency guarantees for
goods. Additionally, note that while Theorem~\ref{thm:: impossbility-deterministic-EQ1} indicates
an impossibility result regarding the notion of EQ1, it remains unknown whether deterministic
mechanisms can achieve strategyproofness, Pareto optimality, and other fairness criteria such as
EF1 or PROP1.

\bibliographystyle{apa}
\bibliography{mylibrary}
%\bibliography{chore-mechanism}

\newpage
\renewcommand{\thesection}{A}\setcounter{equation}{0} \renewcommand{\theequation}{A-\arabic{equation}}
\setcounter{page}{1}\renewcommand{\thepage}{A-\arabic{page}}
\section*{Appendix}

\subsection{For Section~\ref{sec::deterministic}}
\label{Appendix::Sec:deterministic}

\subsubsection{Proof of Theorem~\ref{proposition::SEQ-pick}}

Without loss of generality, we assume agents are ordered $1,\ldots,n$ and hence agent 1 is the
first to pick. Every sequential picking mechanism can be characterized by a sequence of number
$t_1, \ldots, t_n $ with $ \sum_{ i \in [n]} t_i = m$. Moreover, such a sequence is predetermined
and is not affected by reported profiles. Fix a sequence $\{ t_i \}_{i = 1} ^ n $. If $ t_1 < m $,
then consider an instance $\cI_1$ where the type of agent 1 is $ v_1 ( e ) = 0 $ for all $e \in E$
and the type of agent $ i \geq 2$ is $ v_i ( e ) = - 1$ for all $e \in E$. Instance $ \cI_1 $
admits an allocation in which all agents have a value zero. However, in the allocation returned by
a sequential picking algorithm with $ t_1 < m $, there exists at least one agent receiving negative
value, and consequently, the returned allocation is not Pareto optimal. Note that Pareto optimality
requires $ t_1 = m $, while any sequential picking algorithm with $ t_1 = m $ cannot output Pareto
optimal allocation for another instance $ \cI_2 $ where the type of agent 1 is $v_1(e) = - 1$ for
all $ e \in E$ and the type of agent $ i \geq 2$ is $ v_i ( e ) = 0 $ for all $ e \in E$.
Therefore, no sequential picking mechanism can always return PO allocations for both $\cI_1$ and $
\cI_2 $.

\subsection{For Section~\ref{Sec::chores}}
\label{Appendix::Sec:chores}

\subsubsection{Proof of Lemma~\ref{Lemma::GSP_chore2}}

Let ${Q}, \bar{{Q}}$ and ${Q}^{\prime},\bar{{Q}} ^ {\prime}$ be the corresponding index sets
constructed in Step~\ref{ALG::GSP-CHORE-STEP-2} of Algorithm \ref{alg::randomchore} with reported
profiles $(\hat{v}_S , \hat{v}_{ - S } )$ and $ (\hat{v}^ {\prime}_S, \hat{v}_{ - S } )$,
respectively. For every $ i \in S $, define $P_i  = \{ e \in E \mid \hat{v}_i ( e ) = 0 \}$ and $ P
^ {\prime}_i  = \{ e \in E \mid \hat{v} ^ {\prime}_i (e) = 0 \} $. Then, for any $ i \in S$ and $ e
\in P_i $, if $ v_i ( e) = 0 $, by the construction of $ \hat{v} ^ {\prime}_i $, we have $ \hat{v}
^ {\prime}_i  ( e ) = 0 $, implying $ e \in P ^ {\prime}_i $. If $ v_i ( e ) \neq 0 $, again from
the definition of $ \hat{v} ^ {\prime}_i $, we have $\hat{v} ^ {\prime}_i ( e ) = \hat{v}_i ( e ) =
0 $, implying $ e \in P ^ {\prime}_i $. Thus, for any $ i \in S$ and $e \in P_i$, item $e \in
P^{\prime}_i$ always holds, i.e., $ P_i \subseteq P_i ^{\prime} $ for all $ i \in S$. Note that
reported valuations of agents $N\setminus S$ are identical in these two profiles, then we can claim
that ${Q} \subseteq {Q} ^ {\prime}$, and equivalently, $\bar{{Q}} ^ {\prime} \subseteq \bar{{Q}}$.

The expected utility of an agent indeed comes from two parts, the assignment of Step
\ref{ALG::GSP-CHORE-SETP-3} and of Step \ref{ALG::GSP_CHORE-STEP-5}. For any $ i \in S$, let $C_i =
\{ e \in E \mid v_ i (e) \neq 0 \}$ be the set of items of which the valuation is non-zero for
agent $i$. If $ C_i \cap P_i \neq \emptyset$ (resp., $ C_i \cap P ^ {\prime}_i \neq \emptyset$) for
some $i$, then agent $i$ would receive a negative expected utility from the assignment of $C_i \cap
P_i $ (resp., $ C_i \cap P ^ {\prime}_i $) under the reported profile $(\hat{v}_S, \hat{v} _{- S})$
(resp., $(\hat{v} ^ {\prime}_S, \hat{v}_{ - S })$). Recall that for any $i\in S$, $ P_i \subseteq P
^ {\prime}_i $ holds, and thus, $ C_i\cap P_i \subseteq C_i \cap P ^ {\prime}_i $. For any $i\in S$
and $ e \in C_i \cap P ^ {\prime}_i $, we have $\hat{v} ^ {\prime}_i ( e ) = 0 $, and moreover, by
the construction of $ \hat{v} ^ {\prime}_i$, it holds that $ \hat{v}_i (e )= 0 $, implying $ e \in
C_i \cap P_i $. Accordingly, $ C_i \cap P^{\prime}_i \subseteq C_i \cap P_i $ holds for all $ i \in
S$, and thus, $ C_i \cap P_i = C_i \cap P ^ {\prime}_i $ .

We now analyze the expected utility of agent $i \in S$ caused by the assignment of Steps
\ref{ALG::GSP-CHORE-SETP-3} and \ref{ALG::GSP_CHORE-STEP-5}. For Step \ref{ALG::GSP-CHORE-SETP-3},
note that for any $e \in C_  i \cap P_i $, it never happens that $ \hat{v}_j ( e ) = 0 $ but $
\hat{v}^{\prime}_j (e) = v(e)$ for some agent $j \in N$. Then, the number of agents reporting zero
on $e$ under $(\hat{v}^{\prime}_S, \hat{v}_ {-S} )$ is at least that under $(\hat{v}_S, \hat{v}_
{-S} )$. Since $\RandChore$ uniformly at random assigns $e$ to an agent whose reported valuation is
zero, the probability of assigning $e$ to agent $i$ under reported profile $(\hat{v}^{\prime}_ S,
\hat{v}_{ - S })$ is no greater than that under $(\hat{v}_ S, \hat{v}_{ - S })$. Hence, the
expected utility of agent $ i $ caused by the assignment of Step~\ref{ALG::GSP-CHORE-SETP-3} does
not decrease when agents $S$ deviate their reporting from $\hat{v}_S$ to $\hat{v}^{\prime}_S$. As
for Step~\ref{ALG::GSP_CHORE-STEP-5}, by Lemma~\ref{Lemma::GSP1}, the expected utility here of
agent $i \in S$ under reported profiles $(\hat{v}^{\prime}_S, \hat{v}_{ - S})$ and $(\hat{v}_S,
\hat{v}_{ - S})$ is equal to $\frac{1}{n} v_i (\cup_{q \in \bar{{Q}} ^ {\prime}} e_q )$ and
$\frac{1}{n} v_i (\cup_{q \in \bar{{Q}} } e_q )$, respectively. Recall that $\bar{{Q}} ^ {\prime}
\subseteq \bar{{Q}}$, we have $\frac{1}{n} v_i (\bigcup_{q \in \bar{{Q}} ^ {\prime}} e_q ) \geq
\frac{1}{n} v_i (\bigcup_{q \in \bar{{Q}} } e_q )$ as chores yield non-positive valuations.
Therefore, for any agent $ i \in S$, the expected utility of his under $(\hat{v}^{\prime}_S,
\hat{v}_ { - S })$ is at least that under $(\hat{v}_S, \hat{v}_ { - S })$.

\subsubsection{Proof of Lemma~\ref{Lemma::GSP_chore3}}

For the sake of a contradiction, assume that there exist valuation functions $\hat{v}_S \neq v_S $
such that
$$
\sum_{i\in S}\mathbb{ E }_{\mathbf{A} \thicksim \widetilde{\mathcal{M}} ^ *
	(\hat{v}_S, \hat{v} _{ -S})}[ u_i (\mathbf{A}, v_i )] > \sum_{i\in S}\mathbb{ E }_{\mathbf{A} \thicksim \widetilde{\mathcal{M}} ^ *
	({v}_S, \hat{v}_{ -S})}[ u_i (\mathbf{A}, v_i )].
$$
Then we consider valuation functions $\hat{v}^{\prime}_S$ constructed as follows; for any $ i \in
S$ and $e \in E$, if $ v_i ( e ) = 0 $, then set $\hat{v}_i ^{\prime} (  e ) = 0 $; otherwise,
$\hat{v}_i ^{\prime} ( e ) = \hat{v}_i (e)$. By Lemma~\ref{Lemma::GSP_chore2}, we know
$\mathbb{E}_{\mathbf{A} \thicksim \widetilde{\mathcal{M}}^*(\hat{v}^{\prime}_S, \hat{v}_{-S})} [u_i
( \mathbf{A}, v_i )] \geq \mathbb{E}_{\mathbf{A} \thicksim \widetilde{\mathcal{M}}^*( \hat{v}_S,
\hat{v}_{-S})} [ u_i ( \mathbf{A}, v_i )] $ for all $ i \in S $. As a consequence, we can without
loss of generality assume that for every agent $ i \in S$, if $ v_i ( e ) = 0 $, then $ \hat{v}_i (
e ) = 0 $. In other words, reported valuation functions $ \hat{v}_S $ only contain one type of
misreporting; that is, some agent $ j \in S $ reports $ \hat{v}_j ( e ) = 0 $ on chore $e$ while
his true valuation is $ v_j ( e ) = v  ( e ) < 0 $.

Let ${Q}, \bar{{Q}}$ and ${Q}^{b},\bar{{Q}} ^ {b}$ be the corresponding index sets constructed in
Step~\ref{ALG::GSP-CHORE-STEP-2} of Algorithm~\ref{alg::randomchore} under reported profiles $({v}
_ S , \hat{v}_{ - S } )$ and $ (\hat{v}_S, \hat{v}_{ - S } )$, respectively. For every $i \in S $,
define $ P_i = \{ e \in E \mid v_i (e) = 0 \}$ and $ P ^ {b}_i = \{ e \in E \mid \hat{v}_i ( e ) =
0 \}$. Based on the aforementioned assumption, for any $i \in S$ and $ e \in P_i $, it holds that $
e \in  P ^ {b}_i $. As reported valuations of agents $N\setminus S$ are consistent, we have $ {Q}
\subseteq {Q} ^ {b}$, equivalent to $ \bar{{Q}} ^ {b } \subseteq \bar{{Q}}$.

If $\bar{{Q}} ^ {b} = \bar{{Q}}$,
then we have
$$
\begin{aligned}
\sum_{ i \in S }\mathbb{E}_{ \mathbf{A} \thicksim \widetilde{\mathcal{M}}
	^* (v_S, \hat{v}_ { - S } ) } [u_i ( \mathbf{A}, v_i  )] =
\sum_{ i \in S }\frac{1}{n} v_i (\bigcup _{ q \in \bar{{Q}}} e_q) &= \sum_{ i \in S }\frac{1}{n} v_i (\bigcup _{ q \in \bar{{Q}}^{b}} e_q) \\
& \geq
\sum_{ i \in S }\mathbb{E}_{ \mathbf{A} \thicksim \widetilde{\mathcal{M}} ^* (\hat{v}_ S, \hat{v}_{ - S } ) } [u_i ( \mathbf{A}, v_i  )],
\end{aligned}
$$
where the first equality transition is because the assignment of $\bigcup_{ q \in Q} e_q $ results
in expected utility zero for every agent $i \in S$ under the reported profile $ (v_S, \hat{v}_{ -
S})$ and the inequality transition is due to that chores allocated in Step
\ref{ALG::GSP-CHORE-SETP-3} yields non-positive expected utility. The above inequality contradicts
the construction of $\hat{v}_S$.

If $\bar{{Q}} ^ b \subsetneq \bar{{Q}}$, as agents $N \setminus S $ consistently report $ \hat{v}_{
- S }$, every item $e_q$ with $q \in \bar{{Q}} \setminus \bar{{Q}}^b$ must be allocated with
probability one to one or a subset of agents in $ S $ under reported profile $(\hat{v}_S , \hat{v}
_ { - S })$. As the set of indices $\bar{Q}$ corresponds to reported profile $(v _S, \hat{v} _{ -
S})$, then for every chore $e_q $ with $q \in \bar{Q}$, it holds that $ v_i ( e_q) = v (e _q)$ for
all $ i \in S$; note that if $ v_i ( e_q) = 0 $, then $q \in Q$. Hence, under $(v_S, \hat{v} _{ -
S})$, the expected utility of agents $S$ caused by assignment of Step \ref{ALG::GSP-CHORE-SETP-3}
is equal to zero, and we have the following:
$$
\begin{aligned}
	\sum_{i\in S}\mathbb{ E }_{\mathbf{A} \thicksim \widetilde{\mathcal{M}}
		^*({v}_S, \hat{v}_{ -S})}[  u_ i ( \mathbf{A}, v_i  )]
	& = \sum_{ i \in S} \frac{1}{n} v_i ( \bigcup_{q \in \bar{{Q}} ^ b} e_q)  +
	\sum_{ i \in S} \frac{1}{n}  v_i ( \bigcup_{ q \in \bar{{Q}} \setminus\bar{{Q}}^b }e_q )\\
	& =  \sum_{ i \in S} \frac{1}{n} v  ( \bigcup_{q \in \bar{{Q}} ^ b } e_q) + \sum_{ i \in S} \frac{1}{n} v  (\bigcup_{ q \in \bar{{Q}} \setminus\bar{{Q}}^b }e_q )\\
	&\geq \sum_{ i \in S} \frac{1}{n} v  ( \bigcup_{q \in \bar{{Q}} ^ b} e_q) 	+
	\sum_{ q \in\bar{{Q}} \setminus\bar{{Q}}^b  } v ( e_q) \\
	& \geq \sum_{i\in S}\mathbb{ E }_{\mathbf{A}
		\thicksim \widetilde{\mathcal{M}} ^ * (\hat{v}_S,\hat{v}_{ -S})}[  u_i ( \mathbf{A}, v_i )],
\end{aligned}
$$
where the second equality transition is due to $v_i ( e_q)  = v(e_q)$ for all $ i \in S$ and $q \in
\bar{Q}$; the first inequality transition is due to $ |S| \leq n  $ and chores yielding
non-positive valuation; the second inequality is owing to the fact that every item $e_q$ with $ q
\in \bar{{Q}} \setminus \bar{{Q}} ^ b $ is allocated with probability one to one or a group of
agents in $S$ under reported profile $ (\hat{v}_S, \hat{v}_{ - S })$. The above inequality leads to
a contradiction, completing the proof.

\subsubsection{Proof of Proposition~\ref{prop::ex-post and ex-ante fairness of M}}

    Let ${Q}, \bar{{Q}}$ be the sets of indices constructed in
	Step~\ref{ALG::GSP-CHORE-STEP-2} under profile $\mathbf{v} = ( v _1,\ldots, v _ n )$ and let
	  $\mathbf{A} = (A_1,\ldots, A _ n )$
	be the returned randomized allocation.
    Suppose $\mathbf{A}^{\prime} = (A^{\prime}_1,\ldots, A^{\prime} _ n ) = (a^{\prime} _{ j , i }) _ { j \in [m], i \in [n]}$ be the corresponding fractional allocation (matrix) implemented by $\mathbf{A}$.
    As the assignment of Step \ref{ALG::GSP-CHORE-SETP-3} yields an expected utility zero for all agents,
    we the focus on the utility driven by the allocation of Step \ref{ALG::GSP_CHORE-STEP-5}.
	Since permutation $\sigma$ is generated uniformly at random,
    for any agent $i$ and any item $e_q$ with $q \in \bar{Q}$,
    the probability of assigning $e_q$ with agent $ i $ is equal to $\frac{1}{n}$,
    and accordingly,
    $a^{\prime} _ {q, i} = \frac{1}{n}$ for all $q \in \bar{Q}$ and $ i \in [n]$.
    Then, agents' valuation in $\mathbf{A} ^ {\prime}$ is
    $ v _ i ( A ^{\prime} _ i) = v _ i ( A ^{\prime} _ j ) = \frac{1}{n} \cdot \sum_{ q \in \bar{{Q}}} v ( e _ q ) = v _ j ( A ^{\prime} _ j  ) $ holds for all $ i ,j \in [n]$
    since $ v _ l ( e _ q ) = v  ( e _ q)$ for all $q \in \bar{Q}$ and all $l \in [n]$.
    Thus,
    fractional allocation $\mathbf{A} ^ {\prime}$ is EF (and hence PROP) and EQ, and therefore, randomized allocation $\mathbf{A}$ is ex-ante EF (and hence PROP) and EQ.

	For the ex-post fairness guarantee, let $\mathbf{A} ^ * = (A^*_1, \ldots, A^*_n)$
	be an arbitrary deterministic allocation in the support of $\mathbf{A}$.
    We first show that $\mathbf{A} ^* $ is EQ1.
    Note that $ v _ i ( e  _q ) = 0 $ holds for all $ i \in [n]$ and $ q \in Q$,
    then the assignment of $\bigcup _ { q \in {Q}} e _ q $ in Step \ref{ALG::GSP-CHORE-SETP-3} does not affect agents' value.
    Thus, when proving EQ1,
    it suffices to consider only the reduced instance with set of items
    $E = \bigcup _ { q \in \bar{{Q}} } e _ q$ and valuation functions
	$ v _ i ( e) = v ( e) < 0 $ for every $ i \in [n]$ and $ e \in E$.
    Let $\sigma ^ * $ be the permutation in Step \ref{ALG::GSP_CHORE-STEP-5} corresponding the deterministic allocation $\mathbf{A}^*$.
    Let the number of items be $ m = kn + d $ with $k ,d \in \mathbb{N}$ and $ 0 < d \leq n$, and hence, the assignment in Step~\ref{ALG::GSP_CHORE-STEP-5} has $k+1$ rounds.
	Given two agents $i, j \in [n]$, without loss of generality,
	we assume $ \sigma^*(i) < \sigma ^ * ( j ) $.
	If agents $i,j$ receive same number of items in $\mathbf{A}^*$,
	then in every single round,
	agent $ i $ receives value no less than that of agent $j$, implying $ v _ i ( A _i ^* ) \geq v _ j ( A _ j ^*)$.
	As for agent $j$, her value in every round $ l \leq k - 1$ is at least the value received by agent $ i $ in round $ l + 1 $.
	So by eliminating the last chore received by agent $j$,
	the value of agent $j$ is at least that of agent $i$. Thus, for this case, allocation $\mathbf{A}^*$ is EQ1.
	For the situation where agent $i$ receives one more item, similarly the value received by agent $ j $ in every round $ l \leq k$ is at least the value received by agent $ i $ in round $ l + 1 $, implying $v _ j ( A ^ * _ j) \geq v _ i ( A ^* _ i  )$, i.e., agent $j$ satisfies the property of EQ1.
	As for agent $ i $, if the last item she receives is removed,
	she would have a value at least that of agent $ j $.
	Therefore, allocation $\mathbf{A} ^ * $ is EQ1.

    We next prove that $\mathbf{A} ^ * $ is also EF1.
    For any pair of agents $ i ,j \in [n]$,
    we have
    $$
     v _ i ( A _ j ^ * ) = v _ i ( A _ j ^* \cap
     \bigcup  _{ q \in \bar{Q}} e _ q ) + v _ i ( A _ j ^* \cap
     \bigcup  _{ q \in {Q}} e _ q  ) = v ( A_j^* \cap
     \bigcup  _{ q \in \bar{Q}} e _ q )
     = v _j ( A_j^* \cap
     \bigcup  _{ q \in \bar{Q}} e _ q ) = v _ j ( A _ j ^*),
    $$
    where the second and the third equality transitions are due to the fact that for any agent $l$ and any $e_q $ with $q \in \bar{Q}$, $v _ l ( e  _q ) = v(e_q) $ holds.
    As $\mathbf{A}^*$ is EQ1,
    the following holds,
    $$ \max _ { e \in A ^* _ i } v _ i ( A_  i\setminus \{ e \}) \geq v _ j ( A  ^* _ j) = v _ i ( A _ j ^*).$$
    Therefore, allocation $\mathbf{A} ^ * $ is also EF1 (and hence PROP1).

\subsubsection{Proof of Proposition \ref{prop::bobw-efficiency}}

Let ${Q}, \bar{{Q}}$ be the sets of indices constructed by Step~\ref{ALG::GSP-CHORE-STEP-2} under
profile $\mathbf{v} = ( v _1,\ldots, v_n )$ and let $\mathbf{A} = (A_1,\ldots, A_n )$ be the
returned randomized allocation, which implements fractional allocation (matrix)
$\mathbf{A}^{\prime} = (A^{\prime}_1,\ldots, A^{\prime}_n ) = (a^{\prime} _{ j , i })_{ j \in [m],
i \in [n]}$.
According to the proof of Proposition \ref{prop::ex-post and ex-ante fairness of M},
for any agent $i$,
his valuation in $\mathbf{A} ^ {\prime}$
is $ v_ i ( A ^ {\prime}_i ) = \frac{1}{n} \sum _{ q \in \bar{Q}} v ( e_q )$, which implies the utilitarian welfare $ \UW(\mathbf{A} ^ {\prime}) = \sum_{ q \in \bar{Q}} v  ( e _q )$ and egalitarian welfare
$\EW(\mathbf{A} ^ {\prime}) = \frac{1}{n} \sum _{ q \in \bar{Q}} v ( e_q )$.
For an arbitrary fractional allocation $\mathbf{B} = ( B_1,\ldots, B_n )$,
it holds that $\UW(\mathbf{B}) \leq \sum_{q \in \bar{Q}} v ( e_q )$
as $ v_i ( e  _q ) = v (e_q)$ for all $i\in [n]$ and all $q \in \bar{Q}$. Moreover, by the pigeonhole principle,
$\EW(\mathbf{B}) \leq \frac{1}{n} \sum_{q \in \bar{Q}} v ( e_q )$ holds.
Therefore, we can claim that
randomized allocation $\mathbf{A}$ is ex-ante UWM and EWM.
By Proposition \ref{prop::implication-relation},
allocation $\mathbf{A}$ is also ex-ante PO and ex-post UWM and PO.

It remains to show the ex-post guarantee regarding EWM. Let $\mathbf{A} ^ * = (A^*_1, \ldots,
A^*_n)$ be an arbitrary deterministic allocation in the support of $\mathbf{A}$ and $\sigma ^ * $
be the permutation in Step \ref{ALG::GSP_CHORE-STEP-5} corresponding $\mathbf{A} ^ *$. Without loss
of generality, assume $\sigma ^ *(i ) = i$ for all $ i \in [n]$. Similar to the proof of
Proposition~\ref{prop::ex-post and ex-ante fairness of M}, we can further assume $E = \bigcup_{ q
\in \bar{{Q}}} e_q $; note that agents indeed have identical valuation functions over items
$\bigcup_{ q \in \bar{{Q}}} e_q $. Let $ | E | = kn + d $ with $k, d \in \mathbb{N}$ and $ 0 < d
\leq n$, then each agent $ i \in [d]$ receives $ k + 1 $ items and each agent $ i \geq d + 1$
receives $k$ items. Moreover, agent $d$ receives the last item in Step~\ref{ALG::GSP_CHORE-STEP-5}.
We now show
   $ v_d ( A ^*_d) \leq v_i (A ^*_  i )$ for all $ i \in [n]$.
For $ i < d $, both agents $d$ and $i$ receive $ k + 1 $ items and in each round,
  the value of agent $i$ is at least the value of agent $d$, which implies $ v_d ( A ^*_d) \leq v_i (A ^*_  i)$.
	For $ i > d $, agent $ i $ receives an item in rounds $1, \ldots, k$ and moreover for any $l\in [k]$, the value of agent $i$ in every round $ l $ is at least the value of agent $ d $ in round $ l + 1 $.
	Therefore, it holds that $ v_d ( A ^*_d) \leq v_i (A ^*_  i)$ for all $ i \in [n]$ and thus
 $\EW(\mathbf{A} ^ *) = v_d ( A^*_ d )$.
	
Denote by $\OPT_E$ and by $\OPT_U$ respectively the optimal utilitarian and egalitarian welfare among all deterministic allocations.
   We then show $ v_d ( A_d ^* \setminus \{ e ^ {\prime}\} ) \geq \OPT_E$ where $ e ^{\prime} $ is the last item received by agent $d$.
For a contradiction, assume $ v_d ( A_d ^* \setminus \{ e ^{\prime}\} ) < \OPT_E$.
By Proposition \ref{prop::ex-post and ex-ante fairness of M},
   allocation $\mathbf{A}^*$ is EQ1, and thus for any $i\in [n]$, $ v_i ( A_i ^*) \leq \max_{ e \in A ^*_d } v_d ( A ^*_d \setminus \{ e \} ) = v_d ( A ^ *_d \setminus \{ e ^{\prime} \} )$ holds where the equality transition is due to the allocation rule of Step \ref{ALG::GSP_CHORE-STEP-5}.
   Accordingly, we have an upper bound of the optimal utilitarian welfare as follows,
$$
\OPT_U = \UW(\mathbf{A}^*) = \sum_{i \in [n]} v_i ( A ^ *_i ) \leq n \cdot v_d ( A_d ^ * \setminus \{ e ^{\prime} \} ) + v_d ( e ^{\prime}) < n \cdot \OPT_E,
$$
where the first equality transition is because $\mathbf{A} ^*$ is ex-post UWM and the last inequality is due to our assumption $ v_d ( A ^ *_d \setminus \{ e ^{\prime} \}) < \OPT_E$ and $ v_d ( e^{\prime}) \leq 0 $.
However, since agents have identical valuations over $\bigcup_{ q \in \bar{Q}} e_q$, it must hold that $\OPT_U \geq n\OPT_E$, contradicting the above inequality.
   Therefore, the claim $ v_d ( A ^*_d \setminus \{ e ^{\prime} \}) \geq \OPT_E$ is proved.
   Additionally, it is not hard to verify $ v_i ( e ^ {\prime}) \geq \OPT_E$ for all $i \in [n]$ as agents have identical valuations over $\bigcup_{ q \in \bar{Q}} e_q $.
   Then, the following holds,
   $$
\EW (\mathbf{A}^*) = v_d ( A ^ *_d  ) = v_d ( A ^ *_d \setminus \{ e ^ {\prime} \} ) + v_d ( e ^ {\prime} ) \geq 2\cdot \OPT_E,
$$
which completes the proof.

\subsubsection{Proof of Proposition \ref{prop::impo-ewm}}

	Let us consider the instance with $n$ agents and a set $E = \{ e_1, \ldots, e_{ (n-1) n + 1}\}$ of $( n -1) n + 1$ chores.
	Agents have identical valuation functions: $v_i ( e_j ) = v ( e_j )$ for all $i,j$,
 and inherent values are $ v ( e_j ) = -1$
	for all $ j \leq (n-1)n $ and $ v ( e_{n(n-1) + 1} ) = - n $.
	For any item $e$, we say $e$ a $\beta$-\emph{chore} if $v(e)  = -\beta$ for $\beta \in \{ 1, n \}$.
	Thus, there are $n(n-1)$ 1-chore
	and one $n$-chore.
	Denote by $\OPT_E$ the optimal egalitarian welfare
	among all deterministic allocations,
 and naturally, $ \OPT_E \leq - n $ as there must be an agent receiving the $n$-chore.
	Consider an allocation $\mathbf{B}$ in which every agent $ i \in [ n - 1 ] $ receives a number $ n $ of 1-chore
	and agent $n$ receives the unique $n$-chore.
	One can verify that $\min_{i \in [n]} v_i (B_i ) = - n $ and thus $ \OPT_E = - n $.

    Let $\sigma ^ * $ with $\sigma ^ * (i ) = i $ for all $ i \in [n]$ be a deterministic permutation in Step \ref{ALG::GSP_CHORE-STEP-5} and
    $\mathbf{A} ^ * = (A ^*_1, \ldots, A^*_n)$
	be the deterministic allocation returned by $\RandChore$ with respect to $\sigma ^ * $.
    Accordingly, $\mathbf{A} ^ * $ is a deterministic allocation in the support of $\widetilde{\cM} ^*(\mathbf{v})$.
    By the allocation rule of Step \ref{ALG::GSP_CHORE-STEP-5},
    for any $2 \leq i \leq n$,
    agent $ i $
    receives a number $n - 1 $ of 1-chore
	and agent 1 receives a number $n - 1 $ of 1-chore and the unique $n$-chore.
	Thus, it holds that $ \EW(\mathbf{A} ^ *) = v _1 (A ^*_1 ) = -2n + 1 $ and therefore,
 the approximation of $\mathbf{A} ^*$ towards the optimal egalitarian welfare is at least $2 - \frac{1}{n}$, approaching $2$ as $ n \rightarrow + \infty$.

\subsubsection{Example in Freeman et al.}

Consider an indivisible chores instance with four agents and eight chores. Agents have 2-restricted
additive valuations as shown in the following table.

 \begin{table}[ht!]
	    \centering
	  \begin{tabular}{c|cccccccc}
	        & $e_1$ & $e_2$ & $ e  _3$ & $ e  _4$ & $ e  _5$ &  $ e  _6$ & $ e  _7$ &$ e  _8$  \\
	         \hline
	         $v_1(\cdot)$ & $-10$ & $-10$& $-10$& $-10$& $-10$& $-10$& $-10$& $-10$ \\
	         $v_2(\cdot)$ & $-10$ & $-10$& $-10$& $-10$& $-10$& $-10$& $-10$& $-10$ \\
           $ v_3 (\cdot)$ & $ -73$ & $ - 1$& $ - 1$& $ - 1$& $ - 1$& $ - 1$& $ - 1$& $ - 1$ \\
           $ v_4 (\cdot)$ & $ -73$ & $ - 1$& $ - 1$& $ - 1$& $ - 1$& $ - 1$& $ - 1$& $ - 1$
	    \end{tabular}
	\end{table}

For a contradiction, these exists an allocation $\mathbf{A}$ that is EQ1, EF1 and PO. Then, one can
verify that both agents 1 and 2 get exactly one chore in $\mathbf{A}$. Then a total of six chores
are assigned between agents 3 and 4. Assume agent 3 gets at least three chores. Note that at least
one of agents 1 and 2 does not receive $e_1$ in $\mathbf{A}$, and therefore, agent 3 violates EF1
when comparing to that agent. For the detailed proof, we refer readers to
\citep{freemanEquitableAllocationsIndivisible2020b}.

\subsection{For Section~\ref{Sec::MixedItems}}

\subsubsection{Proof of Lemma \ref{Lemma::mechanism-two-players-1}}

Denote by $ \{ Q_i \}_{i=0}^{3}$ and $ \{ Q_i ^ {k} \}_{ i = 0 } ^ {3}$, respectively, the
corresponding indices sets constructed in Step \ref{two-players-mechanism-step-2} under reported
profiles $ (\hat{v}_i , \hat{v}_{3-i})$ and $( \hat{v}^{k}_i , \hat{v}_{ 3 - i } )$. Let $\Delta$
be the difference between expected utility of agent $i$ under reported profile $ (\hat{v}_i ,
\hat{v}_{ 3 - i } )$ and that under $ (\hat{v} ^ k_i , \hat{v}_{ 3 - i } )$, i.e.,
$$
\Delta = \mathbb{E}_{ \mathbf{A} \sim \widetilde{\cM}^2(\hat{v}_i, \hat{v} _{ 3 - i })}
 [ u_i ( \mathbf{A}, v_i )] - \mathbb{E}_{ \mathbf{A} \sim \widetilde{\cM}^2
 (\hat{v} ^ {k}_i, \hat{v} _{ 3 - i })} [ u_i ( \mathbf{A}, v_i )].
$$

In what follows, we prove $\Delta \geq 0$ by carefully checking possible combinations of $\hat{v}_i
( e_k )$, $\hat{v} ^k_i ( e_k ) $ and $ \hat{v}_{ 3 - i}$.

		\medskip\noindent\textbf{Case 1}: $ \hat{v}_i ( e_k ) = v_i ( e_k ) = v ( e_k ) $.
		
		{Subcase 1.1}: $ \hat{v}_i ^ {k} (e_k ) = 0 $ and $ \hat{v}  _{ 3 - i }(e_k ) = v ( e_k ) $.
		Under these reported valuations,
		it holds that
 $ Q_0 = Q ^ k_0 $, $ Q_i = Q^ k_i $, $ Q_{ 3 - i } = Q^ k_{3 - i } \setminus \{ e _k \}$ and $ Q_3 = Q^k_3 \cup \{ e _k \}$.
 Then, by Proposition~\ref{prop::EV-RR-MIX},
 we have $\Delta = \frac{1}{2} v_i ( e_k ) = \frac{1}{2} v ( e_k ) > 0 $.

	{Subcase 1.2}: $ \hat{v}_i ^ {k}( e_k ) = 0$ and $ \hat{v}_{ 3- i } (e_k ) = 0 $.
		Under these reported valuations,
  it holds that $Q_0 = Q ^ k_0 \setminus \{ e _k \}$, $Q_i = Q^k_i \cup \{ e_k \}$, $Q_{ 3 - i } = Q^k_{ 3 - i}$ and $ Q_ 3 = Q^k_3$.
  Similarly, by Proposition~\ref{prop::EV-RR-MIX},
  we have $ \Delta = \frac{1}{2} v_i ( e_k ) = \frac{1}{2} v ( e_k ) > 0 $.
		
		{Subcase 1.3}: $ \hat{v}_i ^ {k} ( e_k ) = 0 $ and $ \hat{v}_{ 3 - i } ( e_k ) = - c ( e_k )$.
		Under these reported valuations,
  the set of indices $ Q_t  $ is identical to $Q^k_t$ for all $ t = 0,1,2,3$, and therefore, $\Delta = 0 $.
		
	{Subcase 1.4}: $ \hat{v} ^ {k}_i (e_k )= - c ( e_k ) $ and $ \hat{v}_{ 3 - i } (e_k ) = v ( e_k ) $.
		Under these reported valuations,
  the composition of $\{ Q ^ {k}_t \}_{ t = 0 } ^ 3$ and of $\{Q_t \}_{ t =0}^3$
		are identical to that of subcase 1.1, and accordingly, we also have $\Delta = \frac{1}{2} v ( e_k ) > 0 $.
		
		{Subcase 1.5}: $ \hat{v} ^ {k}_i ( e_k )= - c ( e_k ) $ and $ \hat{v}_{ 3 - i } (e_k ) = 0 $.
  Under these reported valuations,
  it holds that $Q_ 0 = Q^k_0$, $ Q_i = Q^ k_i \cup \{ e_k \}$, $ Q_{3 - i } = Q^k_{3 - i} \setminus \{ e_k \}$
  and $Q_3 = Q^k_3$.
  Again, by Proposition~\ref{prop::EV-RR-MIX},
  we have $\Delta = v_i ( e_k ) = v ( e_k ) > 0 $.
		
		{Subcase 1.6}: $ \hat{v} ^ {k}_i ( e_k )= - c ( e_k ) $ and $ \hat{v}_{ 3 - i } (e_k ) = - c ( e_k ) $.
  Under these reported valuations,
it holds that $ Q_ 0 = Q^k_0 $, $Q_i = Q^k_i \cup \{  e_k\}$, $Q_{3-i} = Q^k_{3 - i }$ and $Q_3 =
Q_3 ^ k \setminus \{ e_k \}$. According to Proposition~\ref{prop::EV-RR-MIX}, we have $\Delta =
\frac{1}{2} v_i ( e_k ) = \frac{1}{2} v ( e_k ) > 0$.
\medskip

	\noindent\textbf{Case 2}: $ \hat{v}_i ( e_k ) = v_i ( e_k ) = 0$.
 Note that reported valuations $\hat{v}_i$ and $\hat{v} ^ k_i $ only differ on $e_k $.
 According to $\RandMixed$, deviating from $(\hat{v}_i, \hat{v}_{3-i})$ to $(\hat{v} ^k_i, \hat{v}_{3-i})$ only affects indices sets $ Q_t $'s and $ Q^k_t $'s
 on whether includes $e_k$ or not.
As the utility of $e_k$ of agent $i$ is zero, one can verify that for Case 2, $\Delta = 0 $ holds.

\medskip

\noindent\textbf{Case 3}: $ \hat{v}_i ( e_k ) = v_i (e_k ) = - c ( e_k ) $.
		
		{Subcase 3.1}: $ \hat{v}_ i ^ { k } ( e_k ) = v ( e_k ) $ and $ \hat{v}_{ 3 - i } (
e_k) = v ( e_ k ) $. Under these reported valuations, it holds that $Q_0 = Q ^ k_0 $, $Q_i = Q^k_i
$, $Q_{3 - i } = Q^k_{3 - i } \cup \{ e_k \} $ and $Q_3 = Q ^ k_3\setminus \{ e_k \} $. By
Proposition~\ref{prop::EV-RR-MIX}, we have $\Delta = -\frac{1}{2} v_i ( e_k ) = \frac{1}{2} c ( e_k
) > 0 $.

		{Subcase 3.2}: $ \hat{v} ^ {k}_i ( e_k ) = v ( e_k ) $ and $ \hat{v}_{ 3 - i } ( e_k ) = 0 $.
  Under these reported valuations,
  it holds that $Q_0 = Q^ k_0 $, $ Q_i = Q^k_i \setminus \{ e _k \}$, $Q_{3 - i} = Q^ k_i \cup \{ e _k \}$ and $Q_3 = Q^k_3$.
  Then according to Proposition~\ref{prop::EV-RR-MIX},
  we have $\Delta = - v_i ( e_k ) = c(e_k) > 0$.

	{Subcase 3.3}: $ \hat{v}_i ^ {k} ( e_k ) = v (e_k ) $ and $ \hat{v}_{ 3 - i } (e_k ) = -c ( e_k) $.
 Under these reported valuations, it holds that
 $Q_0 = Q^k_0 $, $ Q_i = Q^k_i \setminus \{ e_k \}$, $ Q_{3 - i } = Q^k_{ 3 - i }$ and $Q_3 = Q^k_3 \cup \{ e_ k \}$.
Similarly, by Proposition \ref{prop::EV-RR-MIX}, we have $\Delta = -\frac{1}{2} v_i ( e_ k) =
\frac{1}{2}c ( e_k ) > 0 $.

		{Subcase 3.4}: $
  \hat{v}_ i ^ {k}( e_k ) = 0$ and $ \hat{v}_{3 - i }( e_k ) = v ( e_k ) $.
		Under these reported valuations, we have $ Q^ {k}_t = Q_t $ for all $t=0,1,2,3$, and accordingly, $\Delta = 0 $ holds.

		{Subcase 3.5}: $ \hat{v}_i ^ {k}( e_k ) =  0 $ and $ \hat{v}_{ 3 - i }( e_k ) = 0 $.
Under these reported valuations, it holds that $Q_0 = Q^k_0 \setminus \{ e_k \}$, $ Q_i = Q^k_i $,
$ Q_{3 - i } = Q^k_{3 - i } \cup \{ e_k  \}$ and $Q_3 = Q^k_3$. According to Proposition
\ref{prop::EV-RR-MIX}, we have $\Delta = -\frac{1}{2} v_i ( e_k ) = \frac{1}{2} c ( e_k ) > 0 $.

		{Subcase 3.6}:  $ \hat{v}_i ^ {k}( e_k ) =  0 $ and $ \hat{v}_{ 3 - i }( e_k ) = - c (
e_k )$. Under these reported valuations, it holds that $ Q_0 = Q^k_0$, $ Q_i = Q^k_ i\setminus \{
e_k \}$, $Q_{3 - i } = Q^k_{ 3 - i}$ and $Q_3 = Q ^k_3 \cup \{ e_k \}$. Again, by Proposition
\ref{prop::EV-RR-MIX}, we have $\Delta = -\frac{1}{2} v_i ( e_k ) = \frac{1}{2} c ( e_k ) > 0 $.

\subsubsection{Proof of Proposition \ref{prop::two-players-fairness}}

Denote by $\{ Q_t \} _ {t = 0 } ^ 3$ the set of indices constructed in Step
\ref{two-players-mechanism-step-2} under profile $\mathbf{v} = (v_1, v_2)$ and by $\mathbf{A} =
(A_1,A_2)$ the returned randomized allocation. Suppose $\mathbf{A}^{\prime} = ( A^{\prime} _1,
A^{\prime}_2 )$ be the corresponding fractional allocation (matrix) implemented by $\mathbf{A}$.
According to Proposition~\ref{prop::EV-RR-MIX}, for any agent $ i \in [2]$, he receives a value $
v_i ( A^ {\prime}_i ) = \frac{1}{2} v_i ( \bigcup _{q \in Q_0\cup Q_ 3} e_q) + v_i ( \bigcup_{q \in
Q_i} e_q )$ in fractional allocation $\mathbf{A} ^ {\prime}$. Moreover by arguments similar to that
in the proof of Proposition \ref{prop::EV-RR-MIX}, agent $ i $'s value of $A^{\prime}_{ 3-i}$ is
equal to $ v_i (A^{\prime}_{3-i}) = \frac{1}{2} v_i ( \bigcup_{q \in Q_0 \cup Q_ 3} e_q) + v_i (
\bigcup_{ q \in Q_ {3 - i}} e_q )$. Based on the definitions of $Q_i$ and $Q_{3-i}$, for any
$i\in[2]$ and $q \in Q_i$ (resp., $q\in Q_{3-i}$), we have $ v_i(e_q) \geq 0$ (resp., $v_i(e_q)
\leq 0$). Accordingly, it holds that $ v _i ( \bigcup _{ q \in Q_ i } e_q ) \geq v _i ( \bigcup _{
q \in Q_ {3 - i} } e_q )$, implying $ v_ i ( A_i ^ {\prime}) \geq v_i ( A_{3 - i} ^ {\prime})$.
Thus, fractional allocation $\mathbf{A}^{\prime}$ is EF (and hence PROP), and therefore, randomized
allocation $\mathbf{A}$ is ex-ante EF (and hence PROP).

For the ex-post fairness guarantee, let $\mathbf{A}^* = ( A_ 1 ^*, A_2^*)$ be an arbitrary
deterministic allocation in the support of $\mathbf{A}$. Fix an agent $ i $. As $ v _i ( e_q ) = 0
$ for all $ q \in Q_0$, we have $ v_ i ( A^*_ i ) = v_i ( A^*_ i \cap \bigcup_{q \in Q_i} e_q ) + v
_ i ( A^*_ i \cap \bigcup_{q \in Q_3} e_q )$ and $ v_ i ( A^*_ {3-i} ) = v_i ( A^*_{3-i} \cap
\bigcup_{q \in Q_{3-i}} e_q ) + v_i ( A^*_{3-i} \cap \bigcup_{q \in Q_3} e_q )$. By the definitions
of $Q_ i $ and $Q_{ 3 - i }$, we have $ v _i ( A^*_i \cap \bigcup _{ q \in Q_ i } e_q ) \geq v _i (
A^*_{3-i} \cap \bigcup _{ q \in Q_ {3-i} } e_q )$. As a consequence, it suffices to show that the
partial allocation of items $\bigcup_{q \in Q_3} e_q $ satisfies EF1. Then by arguments similar to
that in the proof of Proposition \ref{prop::ex-post and ex-ante fairness of M}, it is not hard to
verify that the partial allocation upon $\bigcup_{q \in Q_3} e_q $ is EF1, and therefore allocation
$\mathbf{A} $ is ex-post EF1 and PROP1.

\subsubsection{Proof of Proposition \ref{prop::two-players-efficiency}}

According to Proposition \ref{prop::implication-relation}, it suffices to prove only ex-ante UWM.
Let $\{ Q_t \}_{t = 0 } ^ 3$ be the set of indices constructed in Step
\ref{two-players-mechanism-step-2} under profile $\mathbf{v} = (v_1, v_2)$. Denote by $\mathbf{A} =
(A_1,A_2)$ the returned randomized allocation and by $\mathbf{A}^{\prime} = (
A^{\prime}_1,A^{\prime}_2)$ the corresponding fractional allocation implemented by $\mathbf{A}$. It
suffices to show that allocation $\mathbf{A}^{\prime}$ achieves the maximum utilitarian welfare
among all fractional allocations. As two agents have identical valuations on items $ \bigcup_{q \in
Q_0 \cup Q_3} e_q$, the utilitarian welfare derived from the assignment of  $ \bigcup_{q \in Q_0
\cup Q_3} e_q$ must be identical among all (fractional) allocations. For every item $e_q$ with $q
\in Q_1\cup Q_2$, it is allocated to the agent with a larger value in allocation
$\mathbf{A}^{\prime}$, and therefore, we can claim that no fractional allocation can achieve a
utilitarian welfare larger than $\UW(\mathbf{A} ^ {\prime})$.

\end{document}